%% file: 1main.tex
\documentclass[conference]{IEEEtran}

\usepackage{cite}
\usepackage{amsthm}
\usepackage{amsmath}
\usepackage{amssymb}
\usepackage{bm}
\newtheorem{definition}{Definition}
\usepackage[tight,footnotesize]{subfigure}
\usepackage{graphicx}
\usepackage{grffile}
\usepackage{subfigure}
\usepackage{float} 
 \usepackage{mathtools}
 \usepackage{wrapfig}
 \usepackage{afterpage}

\usepackage{lipsum}
\usepackage{cuted}

\usepackage[linesnumbered,ruled]{algorithm2e}

\begin{document}

\newtheorem{theorem}{Theorem}[section]
\newtheorem{conjecture}[theorem]{Conjecture}
\newtheorem{corollary}[theorem]{Corollary}
\newtheorem{proposition}[theorem]{Proposition}

\title{Strategic Server Deployment under Uncertainty in Mobile Edge Computing}
\author{
    \IEEEauthorblockN{
        Duc A. Tran ~~~ Dung  Truong ~~~ Duy Le\\
        Department of Computer Science\\
        University of Massachusetts, Boston, MA 02125\\
        Email: \{duc.tran, dung.truong001, duy.le004\}@umb.edu
        }
}
\maketitle

\thispagestyle{plain}
\pagestyle{plain}

\begin{abstract}
Server deployment is a fundamental task in mobile edge computing: where to place the edge servers and what user cells to assign to them. To make  this decision is context-specific, but common goals are 1) computing efficiency: maximize the amount of workload processed by the edge, and 2) communication efficiency: minimize the  communication cost between the cells and their assigned servers.  We focus on practical scenarios where the user workload in each cell is unknown and time-varying, and so are the effective capacities of the servers. Our research problem is to choose a subset of candidate servers  and assign them  to the user cells such that the above goals are sustainably achieved under the above uncertainties. We formulate  this problem as a stochastic bilevel optimization, which is strongly NP-hard and unseen in the literature. By approximating the objective function with  submodular functions,  we can utilize state-of-the-art greedy algorithms for submodular maximization to effectively solve our problem. We evaluate the proposed algorithm using real-world data, showing its superiority to alternative methods; the improvement can be as high as 55\%.
\end{abstract}

\input{1intro}

\input{1related}

\input{1problem}

\input{1solution}

\input{1eval}
\input{1conclusions}



\bibliographystyle{IEEEtran}

 \bibliography{../bibfiles/tran,../bibfiles/misc,references}

\end{document}

%% file: 1intro.tex
\section{Introduction \label{sec:intro}}
Mobile Edge Computing (MEC) \cite{8030322} is a viable technology for mobile operators to push computing resources closer to the users by deploying edge servers. MEC server deployment consists of two tasks: server placement and server assignment. 
If we consider the placement task alone, seeking a subset of servers from given candidates to optimize some objective,  this objective is usually well-formulated and easy to evaluate. For example, as an instance of the classic Facility Location Problem \cite{Liu2009}, we select $k$ servers from $M>k$ known locations to minimize the average geographic distance between each of $N$ cells and its nearest selected server; to compute this distance for an arbitrary placement is easy. On the other hand, if we consider the assignment task alone, we would be given a set of $k$  servers in advance and seek the best way to assign them to the $N$ cells to maximize the edge benefit. For example, as an instance of the classic Multiple Knapsack Problem \cite{Kellerer2004}, the objective is to maximize the amount of workloads processable by the servers without exceeding their capacity limits.

In the above examples, the two tasks are independent optimization problems. The former assumes that a cell would always be assigned to its nearest server, whereas the latter would not use this assignment because it is not the optimal solution. In practice, we often need to solve these problems jointly, which is the focus of our paper:  select the best $k$ servers among a set of $M$ candidates to optimize an objective that can only be figured out by  solving the assignment problem for these $k$ servers to a set of $N$ cells.  We aim to achieve the following:
(1)  \underline{Computing efficiency}: The total backhaul cost should be minimal. This cost is incurred by the datacenter for processing requests not fulfilled by the edge servers due to edge capacity limit; 
(2) \underline{Communication efficiency}:  Cells should be assigned to servers such that their communication is cost-effective with low latency;
and (3) \underline{Uncertainty robustness}: The server deployment should remain efficient in the presence of time-varying workloads and the dynamics that a server's effective capacity may go up and down in real time. Otherwise we would need to recompute a new deployment frequently to cope with these uncertainties.

\textbf{Challenges. } This joint optimization problem belongs to the class of  bilevel-optimization problems \cite{7942105}: the assignment problem is the  inner optimization task and the placement problem  is the outer optimization task. Unfortunately, bilevel optimization in general is strongly NP-hard \cite{doi:10.1137/0913069} and well-established approximation methods remain lacking to date \cite{7942105}.  

What complicates  our research  is due to the third goal aforementioned: uncertainty robustness. We are aware of no prior MEC research for the following  question: can we find a server deployment solution strategically  that can sustain, without adjustment, a long period of time-varying workloads under time-varying server capacities? Real-time changes in workload and capacity are well-observed in network systems \cite{DBLP:journals/ior/IbrahimW11,Kafetzakis:2011:ESD:1988079.1988274,7362036,TranICCCN19,Modeling15-Chen}.  Combining these uncertainties with the strong NP hardness of bilevel optimization makes our research problem   uniquely challenging.

\textbf{Contributions. }
We call our problem the \underline{S}trategic server \underline{D}eployment under \underline{U}ncertainy problem ($\mathsf{SDU}$).  We make the following contributions.

1) We propose a stochastic bilevel-optimization formulation for $\mathsf{SDU}$. We represent the workload uncertainty    and capacity uncertainty in terms of random variables and set the objective to optimize a stochastic quantity incorporating all the three criteria of computing efficiency, communication efficiency, and uncertainty robustness. From the algorithmic perspective, this formulation is  a new and non-trivial instance of bilevel optimization.   

2) From the application perspective, our research would save costs for real-world implementations. In practice, we do not adjust the server deployment after each individual workload change or individual capacity change. Instead, this adjustment should only take place after a sufficiently long period. We thus enable  deployments that are more sustainable over time. 

3) As $\mathsf{SDU}$ is strongly NP-hard, whose objective function is stochastic, hence computational unfriendly, we propose an approximate algorithm to solve $\mathsf{SDU}$ with polynomial time complexity and bounded approximality. Our approach is based on  sandwiching the bilevel objective function between one-level objectives that are submodular, so that we can utilize state-of-the-art greedy algorithms for submodular maximization to effectively solve our problem. 

4) We evaluate the proposed algorithm using a real-world mobile traffic dataset, showing that not only it is better than intuitive approaches, it can produce a range of server deployments representing arbitrary prioritization of computing efficiency over communication efficiency and vice versa. For example, the evaluation shows that the improvement is 19\% on average and can be as high as 55\% when compared to the best alternative.

The rest of the paper is organized as follows. Related work is discussed in Section \ref{sec:relatedwork}. The $\mathsf{SDU}$ problem is formally introduced with its formulation as a stochastic bilevel optimization  in Section \ref{sec:prob}. The  proposed solution is described in Section \ref{sec:solution}. Numerical results are reported in Section \ref{sec:evaluation}.  The paper concludes in Section \ref{sec:conclusions} with pointers to our future work.

%% file: 1related.tex
\section{Related Work\label{sec:relatedwork}}
A crucial task in MEC design is how to place and assign resources on the edge.  Extensive works have been dedicated to this topic with various objectives and constraints.   For the assignment problem,  Xu et al. \cite{8071527}   aim  to best assign  user workloads and  virtual machines required to run them to the edge servers, subject to server capacity and network delay.  A similar work   in \cite{DBLP:journals/iotj/FanA18} aims to minimize the computational delay, not only network delay.  Geographic constraints are also considered  \cite{TranIJPEDS18}. 
 In \cite{Wang:2017:10.1109/ACCESS.2017.2665971}, where the MEC application consists of inter-dependent tasks organizable into a graph, the goal is to map this task graph to the physical graph of edge servers to minimize the cost to execute the application. 
 
The server placement problem has also been widely studied. Here, we emphasize works that solve this problem jointly with the assignment problem. In \cite{7362036,LAHDERANTA2021130}, the goal is to place among a given set of strategic locations a number of cloudnets to minimize a delay-like cost (network delay \cite{7362036} or geographic distance \cite{LAHDERANTA2021130}) between mobile users and their serving cloudlets, subject to server capacity.  Similarly, Ceselli et al.  \cite{Ceselli:2017:MEC:3148626.3148659} aim to minimize the cloudnet installation cost. Combining both delay cost and installation cost  is the focus of Fan and Ansari \cite{JAS-2018-0416}. Energy consumption can also be incorporated as part of the constraint or objective \cite{10.1007/s11227-023-05692-4}. More recently, Qu et al. \cite{10.1109/TMC.2021.3136868} consider the presence of unreliable edge servers and solve a server placement problem to maximize the expected edge workload while being robust to server failures. 

Our research differentiates from the literature of MEC server deployment in several ways. First, although some existing works, e.g., \cite{7362036,Ceselli:2017:MEC:3148626.3148659,TranICCCN19}, do consider the dynamics of the network state causing  uncertainty in workload and capacity, their approach is to make incremental adjustment upon each state change. Instead, our goal is to compute a   {strategic}  long-lasting solution. Second, the common formulation to incorporate workload changes is trivial: simply approximate them with their average estimate. In our work, we represent dynamic workload and dynamic capacity as stochastic variables, thus better capturing their uncertainty pattern. This makes our mathematical formulation of the problem unique. Third, all solutions in the literature, except some assuming simplistic settings, are based on heuristics.  In contrast, we aim for an approximate solution with bounded near-optimality. In this aspect, the closest work to ours is \cite{10.1109/TMC.2021.3136868} of Qu et al. However, they do not consider the uncertainty factor. Also, their objective function is submodular, whereas the submodularity of our objective function is unknown, hence harder to optimize.

%% file: 1problem.tex
\section{Problem Statement\label{sec:prob}}
\textbf{Cells or base-stations.} Suppose that the service area  consists of  $N$ cells, indexed by the set $[N] = \{1, 2, ..., N\}$. Each cell $i\in [N]$ is served by base-station $i$.  Hereafter, the words ``cell" and ``base-station" are interchangeable.  Each cell $i$ caries a combined workload $W_i$ of computing demanded by its connected mobile devices. As  this workload varies over time, we model it as a random variable $W_i \in [0, \infty)$. 

\textbf{Server candidates.} $M$ candidate servers  are available to join MEC as edge servers,  indexed by the set $[M] = \{1, 2, ..., M\}$. Each server $s \in [M]$ has compute capacity $K_s$. The effective capacity of a server may fluctuate from  time to time due to its dynamic resource usage (the server may involve in other computing tasks). Therefore, we also model $K_s$ as a random variable  $K_s \in [0, \infty)$. In what follows, unless explicitly specified otherwise, we use indices $i$, $j$ $\in [N]$ for referring to cells, and $s, u, v \in [M]$ for servers.

\textbf{Server deployment.} Define a binary matrix  ${z} \in \{0,1\}^{N \times M}$  where $z_{is}$ is 1 iff $i$ is assigned to $s$. A server deployment, denoted by the decision variables $(S, z)$, consists of two tasks: find a subset $S \subset [M]$ to serve as edge servers and an assignment  $z$ to assign the cells to the servers in $S$. 

The objective and constraints for the server deployment are application-specific. In this paper, we assume that we can only deploy $k <<M$ servers and each cell must be assigned to one and only one server; hence, the following constraints:
\begin{align}
|S| &= k, \forall i \in [N]:  \sum_{s \in S} z_{is} = 1  \wedge   \sum_{s  \in [M] \setminus S} z_{is} = 0. 
\label{eq}
\end{align}
We explain the objectives below.

\subsection{Objective 1: Computing Efficiency}
We adopt the edge computing model as in \cite{7134728,TranTransOnTETC21}. Given a deployment $(S, z)$, the workload demand put on server $s \in S$  is
 $\mathcal{W}(s,z) \triangleq \sum_{i=1}^N  z_{is}  W_{i}$.
Because both workload demand and server capacity are uncertain, $W_s$ may exceed server $s$'s capacity  $K_s$. In such a case, the requests that lead to violation of this capacity will be processed by the datacenter in the cloud, incurring a backhaul workload of $\mathcal{W}(s,z) -K_s$. On the other hand, if $K_s > \mathcal{W}(s,z)$, all workload $\mathcal{W}(s,z)$ is processed by server $s$, hence zero backhaul workload.  
We  thus quantify the \textbf{computing efficiency} to be the amount of workload processed by the edge, 
\begin{align}
f(S, z) \triangleq \sum_{s \in S}  \min \bigg (K_s,   \sum_{i \in [N]}  z_{is}  W_{i}  \bigg).
\label{eq:function_f}
\end{align}
Ideally, we want a deployment  $(S, z)$ such that when applied to a (future) time-varying workload this computing efficiency $f(S, z)$  is maximum. 
 
\subsection{Objective 2: Communication Efficiency}
Let $d_{is}$, normalized to range [0, 1], denote the communication cost between cell $i$ and server $s$. In general, this cost can be any measure to represent a pairwise cost of assigning cell $i$ to server $s$; for example, it can be  geographic distance, or network delay, or even some monetary cost according to the service partnership agreement between the operators of the cells and the operators of the edge servers. We then define the communication cost  of a deployment $(S, z)$ to be the sum of the communication cost between a cell and its assigned server, 
\[
\mathsf{communication\_cost} \triangleq   \sum_{s \in \mathcal{S}}\sum_{i \in [N]}   z_{is}d_{is}.
\]
Ideally, we want to minimize this cost, i.e., maximizing 
\begin{align}
g(S, z) &\triangleq N -  \sum_{s \in \mathcal{S}} \sum_{i=1}^N  z_{is} d_{is}
=   \sum_{s \in \mathcal{S}} \sum_{i=1}^N  z_{is} \underbrace{(1- d_{is})}_{\text{denoted by~} c_{is}}
\label{eq:function_g}
\end{align}
which we refer to as the \textbf{communication efficiency}. Here,  $c_{is} \triangleq 1-d_{is}$ can be considered a communication ``score" of assigning cell $i$ to server $s$.

\subsection{Stochastic Bilevel Optimization Formulation}
Because $W_i$'s and $K_s$'s are random variables,  the value of $f(S, z)$  in Eq \eqref{eq:function_f} is also a random variable. To maximize $f$ not knowing future workloads and server capacities, we aim to maximize its expectation,
 \begin{align}
\mathbb{E}\bigg[f(S, z)\bigg] 
=  \sum_{s \in \mathcal{S}} \mathbb{E} \bigg[ \min \bigg(K_s,  \sum_{i\in[N]}  z_{is}  W_{i} )	\bigg].
\label{eq:expected_cost}
\end{align}
We combine the two objectives, $f$ and $g$, into a single objective
\begin{align}
\max_{S, z} ~ \bigg\{ \Omega(S, z)  
\triangleq   \lambda_1 \cdot \mathbb{E}\bigg[f(S, z)\bigg] 
+ \lambda_2  \cdot g(S,z) \bigg\} \label{eq:omega}
\end{align}
where coefficients $\lambda_1 \in [0,1], \lambda_2=1-\lambda_1$ are to tune the priority tradeoff between   computing efficiency and   communication efficiency. 
Taking the constraints   in Eq \eqref{eq} into account,  the deployment problem is equivalent to the following stochastic  bilevel optimization.
\begin{definition} The problem of  \underline{s}trategic server \underline{d}eployment under \underline{u}ncertain workloads and capacities ($\mathsf{SDU}$) is to solve
\begin{align}
&\max_{S,z}  \bigg\{ \Pi(S) \triangleq \Omega(S, z): S \subset [M], |S| = k \bigg\}  \label{eq:leader_task} \\
&\text{subject to~} z = \zeta(S) \text{~where}\nonumber\\  
&\zeta(S) \triangleq \arg \max_{ z \in \{0,1\}^{[N]\times S}}  \Omega(S, z):
\sum_{s\in S} z_{is} = 1 ~ \forall i\in[N]  
\label{eq:follower_task} 
\end{align}
\end{definition}

In the theory of bilevel optimization,   optimization \eqref{eq:follower_task} is called the inner task and   optimization \eqref{eq:leader_task} is called the outer task. We make no specific assumption about the randomness distribution of workloads \{$W_i$\} and capacities \{$K_s$\}. Generally, $\mathsf{SDU}$ is a non-trivial instance of bilevel optimization (strongly NP-hard).  Adding the stochastic factor makes it even harder. The literature of bilevel optimization, particularly its stochastic version, lacks solutions that can generalize widely \cite{BECK2023401}. To copy with stochastic uncertainty, the current approach is to assume a sampling-based estimator for the stochastic term (mean, variance) to reduce to non-stochastic optimization. Our problem cannot apply this because that would result in solving many NP-hard sub-problems.

%% file: 1solution.tex
 \section{The Solution \label{sec:solution}}
 
 $\mathsf{SDU}$ is a maximization of a set function.  Most developments in the area of subset optimization have connections to submodular functions. Submodular maximization is NP-hard but submodular functions exhibit properties that are useful to designing efficient greedy algorithms. These algorithms assume that to evaluate the objective function on a candidate subset takes polynomial time. 
 
Unfortunately,  it is impossible to evaluate our objective function precisely. The expectation term $\mathbb{E}[f(S, z)]$ (Eq \eqref{eq:expected_cost}) in $\Pi$ has no explicit mathematical form that can be computed. Even in cases where we can obtain an explicit formula, for example, by assuming (or approximating) $W_i$ and $K_s$ with certain probability distributions, to evaluate $\Pi$ is still NP-hard due to the inner optimization. 

We propose an approximate algorithm that works for arbitrary workload and capacity distributions. The proposed approach is to sandwich the implicit objective function with explicit lower-bound and upper-bound functions, $\Pi_l$ and $\Pi_u$, that are submodular and polynomial-time, so that we can design a polynomial-time greedy algorithm to solve $\mathsf{SDU}$ with a bounded near-optimality.    In what follows, we provide basic background on submodular maximization and later the details of our algorithms.

\subsection{Submodular Maximization\label{sec:submodular}}

Consider a set function $f: 2^V \mapsto \mathbb{R}$ on subsets of a given finite set $V$ (called the ground set). For convenience, denote $S+v \triangleq S \cup \{v\}$.

\textbf{Submodularity.} $f$ is   {submodular} iff for any subset $S  \subset V$ and any elements $v_1, v_2 \in V$ we have
$f(S + v_2) - f(S) \ge f(S + v_1+v_2)-f(S+v_1)$. That is, the marginal contribution of any element  to the value of a set diminishes as the set grows. Function $f$ is said to have a diminishing return.

\textbf{Monotone. } $f$ is  {monotone} iff for any subset $S \subset V$ and any element $v \subset V$ we have
$f(S) \le f(S+v)$. That is, growing a set can only increase its value.

\textbf{Matroid. }  A matroid, denoted by $M = (V, \mathcal{I})$, is family $\mathcal{I}$ of independent sets that are subsets of a given set $V$ such that  1) $\emptyset \in \mathcal{I}$, 2) \textit{hereditary property}: if $T \in \mathcal{I}$ then $S \in \mathcal{I}$ for every $S \subset T$, and 3) \textit{augmentation property}: for any $S, T \in \mathcal{I}$ with $|S| < |T|$ then $S+v \in \mathcal{I}$ for some $v \in T$. For example, $\mathcal{I} = \{S \subset V: |S| \le k\}$ is a matroid representing a cardinality constraint; it is called the uniform matroid.

\textbf{Submodular maximization. } The problem is to find a subset $S$ of a ground set $V$ that maximizes a submodular function $f(S)$ subject to some set constraint $\mathcal{I} \subset 2^V$:
$\max_S  \{  f(S): S \in \mathcal{I}  \}$.
Constraint $\mathcal{I}$ usually satisfies the hereditary property. Most often, it is a matroid constraint such as the cardinality constraint, or a  knapsack constraint $\mathcal{I} = \{S \subset V : \sum_{s \in S} w_s \le b\}$.   

\textbf{Greedy algorithm. }There is a standard greedy algorithm to solve this problem, proposed by Nemhauser et al.  \cite{10.1007/BF01588971}: start with the empty solution set $S^{(0)} = \emptyset$ and  augment it by adding a new element iteratively, $S^{(t+1)}=S^{(t)}+v^*$ in each step $t$, where
$v^* = \arg \underset{v \in V\setminus S^{(t)}, S^{(t)}+v \in \mathcal{I} }{\max}  \bigg \{f(S^{(t)}+v)-f(S^{(t)}) \bigg\}$.  
This algorithm, denoted by $\mathsf{GREEDY}(f)$, runs in $O(Mk)$ time, where the time unit is the time to evaluate function $f$. Its approximation ratio is $(1-\frac{1}{{e}})$ for the cardinality constraint (which is optimal \cite{10.1145/258533.258641}) and $\frac{1}{2}$ for the matroid constraint. 

\subsection{Submodular Approximation}
Now we are ready to tackle our problem. Let   $[M]$, the set of server candidates, be the ground set and the constraint be the cardinality constraint, $\mathcal{I} = \{S \subset [M]: |S| \le k\}$. The set function to maximize is our objective function $\Pi(S)$:
\[
\Pi(S) \triangleq \Omega(S, \zeta(S)) =  \lambda_1 \cdot  \mathbb{E}\bigg[f(S, \zeta(S))\bigg] 
+ \lambda_2 \cdot g(S, \zeta(S)) 
\]
where 
\[
  \zeta(S) \triangleq \arg \max_{ z \in \{0,1\}^{[N]\times S}} \bigg\{ \Omega(S, z) :
\sum_{s\in S} z_{is} = 1, ~ \forall i\in[N]  \bigg\} .
\]

We can prove that function $\Pi$ is monotone. Indeed, consider any $S \subset [M]$ and let $z = \zeta(S)$ be its optimal assignment. Now, suppose that we are adding server $v \in [M]$ to $S$. Let $s_1$ be an assignment for subset $(S+v)$ such that it assigns no cell to server $v$ and keeps the same assignment $z$ for all servers in $S$. Thus, the efficiency of deployment $(S+v, z_1)$ is the same as that of $(S, z)$. On the other, assignment $s_1$ is no better than the optimal assignment of subset $(S+v)$. So, 
\begin{align*}
\Pi(S+v) &= \Omega(S+v, \zeta(S+v)) \\
&\ge \Omega(S+v, z_1) = \Omega(S, z) = \Pi(S).
\end{align*}
This is true for arbitrary $S$ and $v$. Therefore, $\Pi$ is monotone. 

The submodularity of $\Pi$ is unknown. Even if $\Pi$ is indeed submodular, we cannot evaluate it due to the computationally implicit expectation term (for general distributions of workload and capacity). To overcome, our approach is to sandwich it between two polynomial-time submodular functions thereby we can design an efficient greedy algorithm with bounded approximation ratio.  This is guided by the Sandwich Theory  \cite{GAO202023} which proves the existence of such a sandwiching for any non-submodular function. The challenge is how to find these bounding functions that are not only tight but  easy to compute.

\textbf{Sandwich approximation. } Consider a non-negative non-submodular function $f$ that is bounded by two non-negative submodular functions, $f_1(S) < f(S) < f_2(S)$ for all $S \in \mathcal{I}$. Let  $U_1 = \mathsf{GREEDY}(f_1)$ and $U_2 = \mathsf{GREEDY}(f_2)$ be the solution sets resulted from applying the standard greedy algorithm, $\mathsf{GREEDY}$, to objective functions $f_1$ and $f_2$, respectively. The proposed solution for maximizing $f$ will be
\begin{align*}
U^* = \arg \max \bigg \{f(U_1), f(U_2) \bigg \}.
\end{align*}
This solution offers an approximation ratio  
\begin{align}
\frac{f(U^*)}{f(U^{opt})} \ge \gamma \cdot \underbrace{ \max   \bigg(\frac{f(U_2)}{f_2(U_2)}, \frac{f_1(U^{opt})}{f(U^{opt})} \bigg)}_{1-\epsilon}.
\label{eq:sandwichbound}
\end{align}
where $U^{opt}$   denotes  the optimal solution set of $f$ and $\gamma$ denotes the  approximation ratio of $\mathsf{GREEDY}$. For example, for monotone submodular functions, $\gamma =  1-\frac{1}{{e}}$   if $\mathcal{I}$ is the cardinality constraint or $\gamma =  \frac{1}{2}$ if $\mathcal{I}$ is the matroid constraint.  

The approximation ratio of $U^*$ is only worse than the optimal bound $\gamma$ by a factor $(1-\epsilon)$ which depends on the closeness of the sandwich bounding. We should find the lower-bound submodular   $f_1$ and the upper-bound submodular   $f_2$ that are as close to $f$ as possible. 

\begin{proposition} Let $\kappa_s = \mathbb{E}[K_s]$, $\mu_i = \mathbb{E}[W_i]$, and $\mu = \sum_{i\in[N]} \mu_i$. Notation $\mathbb{V}$ denotes the variance of a random variable. For every deployment ( $S$,  $z$), we have
\begin{align}
f_{l1}(S, z) \le \mathbb{E}[f(S, z)]  \le f_{u1}(S, z) 
\label{eq:bound1}
\end{align}
where
\begin{align}
f_{u1}(S, z) =  \sum_{s \in \mathcal{S}}  \min \bigg( \kappa_s,    \sum_{i\in[N]} z_{is}  \mu_i \bigg) \label{eq:upbound1}\\
f_{l1}(S, z) = f_{u1}(S,z) -\sum_{s \in \mathcal{S}}  \sqrt{ \frac{1}{2} \mathbb{V}\bigg[K_s - \sum_{i\in[N]}  z_{is}  W_{i} \bigg]}.\label{eq:lowbound}
\end{align}
Furthermore, these bounds are tight.
\label{proposition:bound1}
\end{proposition}
\begin{proof}
For any independent random variables $X$ and $Y$,
\begin{align}
\mathbb{E} \bigg[ \min(X, Y) \bigg] \le \min \bigg(\mathbb{E}[X], \mathbb{E}[Y] \bigg),
\label{eq:Emin1}
\end{align}
 and so
\begin{align*}
\forall S, z: ~\mathbb{E}[f(S, z)] 
   &= \sum_{s \in \mathcal{S}}  \mathbb{E} \bigg[ \min \bigg(K_s,  \sum_{i\in[N]}  z_{is}  W_{i} )\bigg]\\
   &\le 
    \sum_{s \in \mathcal{S}}  \min \bigg( \mathbb{E}[K_s],  \mathbb{E}\bigg[  \sum_{i\in[N]} z_{is}  W_{i}\bigg]\bigg)\\
   &=
    \sum_{s \in \mathcal{S}}  \min \bigg( \kappa_s,    \sum_{i\in[N]} z_{is}  \mu_i \bigg).
   \end{align*}
Hence,   the upper bound $f_{u1}(S,z)$ is obtained.

Let us work on the lower bound. For any independent random variables $X$ and $Y$, we have (Corollary 2.1 of \cite{10.2307/3213876}):
\begin{align}
\mathbb{E} [ \min(X,Y) ] \ge \min  (\mathbb{E}[X], \mathbb{E}[Y] ) - \sqrt{\frac{1}{2}   \mathbb{V}[X-Y]}.
\label{eq:Emin2}
\end{align}
Letting $X=K_s$ and $Y=\sum_{i\in[N]}  z_{is}  W_{i}$, we have
\begin{align*}
&\mathbb{E}\bigg[\min \bigg(K_s,  \sum_{i\in[N]}  z_{is}  W_{i} )\bigg] \\
&\ge 
\min  \bigg( \mathbb{E}[K_s],  \mathbb{E}[ \sum_{i\in[N]}  z_{is}  W_{i}] \bigg ) 
- \sqrt{ \frac{1}{2} \mathbb{V} [K_s - \sum_{i\in[N]}  z_{is}  W_{i} ]}\\
&= \min \bigg( \kappa_s,  \sum_{i\in[N]}  z_{is}  \mu_{i}\bigg ) 
- \sqrt{ \frac{1}{2} \mathbb{V}[K_s - \sum_{i\in[N]}  z_{is}  W_{i}]}.
\end{align*}  
Summing both sides over $s \in \mathcal{S}$,  
\begin{align*}
&\mathbb{E}[f(S, z)] 
= 
 \sum_{s \in \mathcal{S}}  \mathbb{E} \bigg[ \min \bigg(K_s,  \sum_{i\in[N]}  z_{is}  W_{i} )	\bigg]\\
&\ge
\sum_{s \in \mathcal{S}} \bigg[ \min  \bigg( \kappa_s,    \sum_{i\in[N]}  z_{is}  \mu_{i}  \bigg )
- \sqrt{ \frac{1}{2} \mathbb{V}[K_s - \sum_{i\in[N]}  z_{is}  W_{i}]} \bigg].
\end{align*}
This is the lower bound $f_{l1}(S,z)$.
Finally, as our inequalities are based on Eq \eqref{eq:Emin1} and Eq \eqref{eq:Emin2} which are tight, the lower bound $f_{l1}$ and the upper bound $f_{u1}$ are also as tight.
\end{proof}

\begin{corollary} For every deployment ( $S$,  $z$), we have
\begin{align}
\mathbb{E}[f(S, z)]  \le f_{u}(S) \triangleq  \min  \bigg(   \mu,  \sum_{s\in\mathcal{S}}\kappa_s\bigg). \label{eq:fu}
\end{align}
\label{proposition:fu}
\end{corollary}
\begin{proof}
According to the upper-bound inequality \eqref{eq:bound1} of Proposition \ref{proposition:bound1}, we have
\begin{align*}
\forall S, z: ~\mathbb{E}[f(S, z)] 
   \le f_{u1} =
    \sum_{s \in \mathcal{S}}  \min \bigg( \kappa_s,    \sum_{i\in[N]} z_{is}  \mu_i \bigg) \\
   \le 
   \min \bigg( \sum_{s \in \mathcal{S}}   \kappa_s,  \sum_{s \in \mathcal{S}}   \sum_{i\in[N]} z_{is}\mu_i\bigg)\\
    = 
     \min \bigg( \sum_{s \in \mathcal{S}}   \kappa_s,   \sum_{i\in[N]} \mu_i \underbrace{\sum_{s \in \mathcal{S}} z_{is}}_{=1}\bigg) = \min   \bigg(   \mu,   \sum_{s\in\mathcal{S}}\kappa_s\bigg). 
\end{align*}
Hence, the upper bound $f_{u}$ is obtained.
\end{proof}
Note that upper bound $f_{u}$ is weaker than $f_{u1}$ but $f_{u}$ does not involve parameter $z$.

\begin{proposition}
For every deployment ( $S$,  $z$), we have 
\begin{align}
g(S, z ) \le g_u(S) \triangleq \sum_{i\in[N]} \max_{s\in S} c_{is}.
\label{eq:gu}
\end{align}
This  bound is precisely the Facility Location objective function.
\label{proposition:gu}
\end{proposition}
\begin{proof}
Consider a cell $i$. According to assignment $z$, suppose that cell $i$ is assigned to server $v$. We have 
$\sum_{s\in S} z_{is} c_{is} = c_{iv} \le  \max_{s\in S} c_{is}$.
So,
$g(S,z) = \sum_{i\in[N]} \sum_{s\in S} z_{is} c_{is} \le \sum_{i\in[N]} \max_{s\in S} c_{is}$.
\end{proof}

\begin{theorem}
For every subset $S$, we have
\begin{align}
\Pi_l(S) \le \Pi(S) \le \Pi_u(S)
\end{align}
where
\begin{align}
&\Pi_u(S) 
\triangleq
 \lambda_1 \min  \bigg(   \mu,  \sum_{s\in\mathcal{S}}\kappa_s\bigg) +\lambda_2 \sum_{i\in[N]} \max_{s\in S} c_{is}\label{eq:PhiU}\\
&\Pi_l(S)
\triangleq
\max_{s\in S} \bigg\{
\lambda_1 \bigg(\min  (\kappa_s,  \mu   ) - \frac{\nu_s}{\sqrt{2}}  \bigg) + \lambda_2 \sum_{i \in [N]} c_{is}\bigg\}.
\label{eq:PhiL}
\end{align}
Here,
\begin{align}
\nu_s \triangleq  \sqrt{\mathbb{V}[K_s -\sum_{i \in [N]} W_i]}
\label{eq:nus}
\end{align}
denotes the standard deviation of the difference between the capacity of server $s$ and the workload total over all the cells.
\label{theorem:main}
\end{theorem}
\begin{proof}
We have
$\Pi(S) = \lambda_1 \mathbb{E}[f(S, \zeta(S))] + \lambda_2 g(S, \zeta(S)) \le \lambda_1 f_{u}(S) + \lambda_2 g_u(S)$
because this is a straightforward application of Corollary \ref{proposition:fu} and Proposition \ref{proposition:gu}. Plugging in the formulas for $f_u$ (Eq \eqref{eq:fu}) and $g_u$ (Eq \eqref{eq:gu}), we obtain the upper bound $\Pi_u$ in Eq \eqref{eq:PhiU}.

Now, let us work on the lower bound. Let $\hat{z}$ be any assignment solution applied to the set $S$ that assigns all the cells to the same server; say, $\hat{s} \in S$ is this server, i.e., 
$\hat{z}_{is} = 1 ~\text{iff}~   s = \hat{s}$.
Then,
\begin{align*}
f(S, \hat{z}) &=   \min \bigg(K_{\hat{s}},  \sum_{i \in [N]} W_i \bigg)\\
g(S, \hat{z})  &=  \sum_{s \in \mathcal{S}} \sum_{i \in [N]}  \hat{z}_{is} c_{is}
= \sum_{i \in [N]} c_{i\hat{s}},
\end{align*}
leading to
$\Pi(S) 
= \Omega(S, \zeta(S)) \ge \Omega(S, \hat{z}) = \lambda_1 \mathbb{E}[f(S, \hat{z})] + \lambda_2 g(S, \hat{z})
= \lambda_1 \mathbb{E}\bigg[ \min \bigg(K_{\hat{s}},  \sum_{i \in [N]} W_i \bigg) \bigg] + \lambda_2 \sum_{i \in [N]} c_{i\hat{s}}$.
About the first term,
\begin{align*}
&\mathbb{E}\bigg[ \min \bigg(K_{\hat{s}},  \sum_{i \in [N]} W_i \bigg) \bigg]\\
&\ge 
\min \bigg(\mathbb{E}[K_{\hat{s}}], \mathbb{E} \sum_{i \in [N]} W_i  \bigg) 
-   \sqrt{\frac{1}{2}   \mathbb{V}\bigg[K_{\hat{s}} -\sum_{i \in [N]} W_i \bigg]}  \\
&=
\min  (\kappa_{\hat{s}},   \mu   )  - \frac{\nu_{\hat{s}}}{\sqrt{2}},
\end{align*}
and so,
\begin{align*}
\Pi(S) \ge
\lambda_1 \bigg(\min  (\kappa_{\hat{s}},   \mu)  \bigg) - \frac{\nu_{\hat{s}}}{\sqrt{2}} \bigg) + \lambda_2 \sum_{i \in [N]} c_{i\hat{s}}.
\end{align*}
Because the above is true for every $\hat{s}\in S$, we have
\[
\Pi(S) \ge \max_{s\in S} \bigg\{
\lambda_1 \bigg(\min  (\kappa_s,  \mu   ) - \frac{\nu_s}{\sqrt{2}} \bigg) + \lambda_2 \sum_{i \in [N]} c_{is} \bigg\}.
\]
This is the lower bound $\Pi_l(S)$ in Eq \eqref{eq:PhiL}.
\end{proof}
Theorem \ref{theorem:main} is significant. 
First, the lower bound $\Pi_l(S)$ and upper bound $\Pi_u(S)$ are explicit mathematical expressions, which do not involve the computation of the expectation term.
Second, they can be evaluated easily based on the input set $S$ only, without involving assignment variable $z$. In other words, maximizing $\Pi_l$ and $\Pi_u$ is a single-level set optimization, whereas $\Pi(S) = \Omega(S, z^*)$  is bilevel which requires the inner optimization finding $z^* = \zeta(S)$.
 Third, crucial for our approach, $\Pi_l$ and $\Pi_u$ are monotone submodular  as we will show below. We thus can leverage linear-time greedy algorithms to obtain near-optimality for our solution.
\begin{proposition}
Function $f_{u}$ is monotone submodular.
\label{proposition:fu2}
\end{proposition}
\begin{proof}
$f_u$ is monotone because for every $v \in [M]$,
\[
f_u(S)  = \min  \bigg(   \mu,  \sum_{s\in\mathcal{S}}\kappa_s\bigg) \le \min  \bigg(   \mu,  \sum_{s\in\mathcal{S}}\kappa_s + \kappa_v \bigg) = f(S+v).
\]
To prove submodularity, for every $v_1, v_2 \in [M]$,
\[
\underbrace{f_{u}(S+v_1+v_2)}_{A} - \underbrace{f_u(S+v_1)}_{B} \le \underbrace{f_u(S+v_2)}_{C} - \underbrace{f_u(S)}_{D},
\]
we divide into three cases below. Denote $\kappa(S) = \sum_{s\in S}\kappa_s$. Also, thanks to monotonicity, we always have $A \ge B$,  $C \ge  D$.
\begin{enumerate}
\item $\mu \le \kappa(S)+\kappa_{v_1}$: $A=B=\mu$. So $A-B = 0 \le C-D$. 
\item  $\mu \ge \kappa(S) + \kappa_{v_1} + \kappa_{v_2}$:  
$A=\kappa(S)+\kappa_{v_1} + \kappa_{v_2}$,
$B=\kappa(S)+\kappa_{v_1}$,
$C=\kappa(S)+ \kappa_{v_2}$, and
$D=\kappa(S)$. So $A-B = C-D =  \kappa_{v_2}$.
\item  $\mu \in (\kappa(S)+\kappa_{v_1}, \kappa(S) + \kappa_{v_1} + \kappa_{v_2})$:
$A=\mu$, $B=\kappa(S)+\kappa_{v_1}$, $D = \kappa(S)$. So, $A-B+D = \mu - \kappa_{v_1} < \mu$. We also have
$ \mu - \kappa_{v_1} < 0 \kappa(S) + \kappa_{v_1} + \kappa_{v_2} - \kappa_{v_1} = \kappa(S) + \kappa_{v_2}$. Thus, $A-B+D \ge \min (\mu, \kappa(S) + \kappa_{v_2}) = C$, meaning $A - B \le C- D$.
\end{enumerate}
The submodularity condition is thus true for all cases. 
\end{proof}
\begin{theorem}
Upper-bound function $\Pi_{u}$ and lower-bound function $\Pi_l$ are monotone submodular.
\label{theorem:PhiSubmodula}
\end{theorem}
\begin{proof}
Consider the upper bound $\Pi_{u}$. We have
\[
\Pi_u(S) =
 \lambda_1 \underbrace{\min  \bigg(   \mu,  \sum_{s\in\mathcal{S}}\kappa_s\bigg)}_{f_u(S)} +
\lambda_2 \underbrace{\sum_{i\in[N]} \max_{s\in S} c_{is}}_{g_u(S)}
\]
According to Proposition \ref{proposition:fu2},  $f_{u}(S)$ is monotone and submodular. Function $g_u(S)$ is the Facility Location objective function which we already know is monotone and submodular. $\Pi_u$ is a linear combination of two monotone submodular functions. Therefore, it is monotone submodular. 

Now, consider the lower bound $\Pi_l$. We have
$\Pi_l(S) = \max_{s\in S} h(s)$
where
\[
h(s) \triangleq  
\lambda_1 \bigg(\min  (\kappa_s,  \mu   ) -  \frac{\nu_s}{\sqrt{2}}   \bigg) \nonumber\\
+ \lambda_2 \sum_{i \in [N]} c_{is}.
\]
Note that $h(s) \ge 0$ for every $s$ because the first term of $h$ is always non-negative.
Function $\Pi_l(S)$ is monotone due to the following inequality for every  $v \in [M]$:
\[
\Pi_l(S) = \max_{s\in S} h(S) \le \max\bigg(h(v), \max_{s\in S} h(s)\bigg) = \Pi_l(S+v). 
\]
The submodularity condition is, for every $v_1, v_2 \in [M]$, 
\[
 \Pi_l(S+v_1+v_2) -  \Pi_l(S+v_1) \le  \Pi_l(S+v_2) -  \Pi_l(S).
 \]
 Let $m \triangleq \Pi_l(S)$. This condition is equivalent to
\begin{align*}
\underbrace{\max(h(v_1), h(v_2), m)}_{A} - \underbrace{\max(h(v_1), m)}_{B} 
\le \underbrace{\max(h(v_2), m)}_{C} - m.
\end{align*}
There are two cases: 
\begin{enumerate}
\item $m \ge h(v_2)$: $A-B = \max(h(v_1),m)-\max(h(v_1), m) = 0 \le C-m$.
\item $m < h(v_2)$: we have
\begin{align*}
&A-B\\ 
&= \max(h(v_1), h(v_2)) - \max(h(v_1), m) \\
&= h(v_1)+h(v_2)-\min(h(v_1), h(v_2)) \\
&~~~- (h(v_1)+m - \min(h(v_1),m))\\
&= h(v_2)-m - \underbrace{(\min(h(v_1), h(v_2))-\min(h(v_1),m))}_{\ge 0}\\
&\le h(v_2)-m = C-m.
\end{align*}
\end{enumerate}
Thus, the submodularity condition is always true for $\Pi_l$.
\end{proof}
 
 \subsection{The Proposed Algorithm}
Because our objective function $\Pi$ is bounded between two monotone submodular functions, $\Pi_l$ and $\Pi_u$, we apply  sandwich approximation  to obtain our approximate solution:
\begin{align}
&S_l = \mathsf{GREEDY}(  \Pi_l)\label{eq:Sl}\\
&S_u = \mathsf{GREEDY}( \Pi_u)\label{eq:Su}\\
&S^*= \arg \underset{S \in \{S_l, S_u\}}{\max}   \{ \Pi(S_l), \Pi(S_u) \}.
\label{eq:Sstar}
\end{align}
Here, $\mathsf{GREEDY}$ is the standard greedy method for submodular maximization presented earlier in Section \ref{sec:submodular}. Functions $\Pi_l(S)$ and   $\Pi_u(S)$ are easily evaluated given $S$, and so it is a straightforward implementation of  $\mathsf{GREEDY}$ to find $S_l$ (Eq \eqref{eq:Sl}) and $S_u$ (Eq \eqref{eq:Su}). 

However,   two challenges remain. The first is due to the need to evaluate the expectation term $\mathbb{E}[f(S, z)]$ of $\Omega(S,z)$ in $\Pi(S)$. If we approximate this expectation by sampling the workload variables $W_i$ and capacity variables $K_s$, the number of  times   needed to solve the above assignment optimization would multiple more. We cannot afford this in practice. Instead, to evaluate $\Pi(S)$ we propose to approximate the term $\mathbb{E}[f(S, z)]$ such that random variables $K_s$ and $W_i$ are approximated by their means, $\kappa_s$ and $\mu_i$, respectively:
\begin{align}
&\mathbb{E}[f(S,z)] = \hat{f}(S,z) \triangleq
\sum_{s \in S} \min (  \kappa_s,  \sum_{i\in[N]} z_{is} \mu_i   ).
\label{eq:Ef}
 \end{align}
Secondly, it is NP-hard to find the final solution $S^*$ in Eq \eqref{eq:Sstar}. This is because we need to evaluate function $\Pi$,
\[
\Pi(S) = \Omega(S, z^*), \text{where~} z^* = \zeta(S) = \arg \max_z \Omega(S, z),
\]
which requires solving the inner optimization to find the best assignment $z^*$ for the servers in $S$; this optimization  is NP-hard. 
To avoid this, we propose an iterative assignment optimization algorithm as follows.

 In each iterative step $t$ of applying $\mathsf{GREEDY}$, we keep track of the best assignment $z^{(t)}$ corresponding to $S^{(t)}$, which is found in the previous iteration. Consider iteration $(t+1)$  and candidate server $v$. To find the best assignment for the would-be augmented set $S^{(t)}+v$, we start with the assignment $z = z^{(t)}$ and conduct a series of greedy steps of reassign a cell to the new server $v$ to continually improve the objective function. 
 
Let $\mathcal{W}_s$ denote the current workload put on server $s$. For $v$, this workload is zero initially. Then we will repeatedly find a cell $j$ to reassign from its current server $u$ to server $v$ if maximizing the following gain (the difference between the efficiency of the assignment before and after):
\[
j = \arg \max_i  ~\bigg\{ \Delta(i)  = \lambda_1 \Delta_f(i) + \lambda_2 \Delta_g(i)  \bigg\}
\]
where
 \begin{align}
  \Delta_f(i) &=  \underbrace{\mathbb{E}[\min(K_v, \mathcal{W}_v + W_i)] +  \mathbb{E}[\min(K_u, \mathcal{W}_u - W_i)]}_{after} \nonumber\\
  &~~-  \underbrace{\mathbb{E}[\min(K_u,\mathcal{W}_u)]}_{before} \label{eq:deltaf}
  \end{align}
  and 
 $\Delta_g(i) = \underbrace{c_{iv}}_{after}-\underbrace{c_{iu}}_{before}$.
 Note that once this cell $j$ is found,  the terms $\mathcal{W}_u, \mathcal{W}_v$ in formula $\Delta_f$ are updated  according to 
 \[
  \mathcal{W}_u \leftarrow \mathcal{W}_u- W_j, ~ \mathcal{W}_v \leftarrow \mathcal{W}_v+ W_j.
  \]
We go on to reassign the next cell repeatedly until there is no more gain. So, at the end, for each candidate $v$ we would have a corresponding would-be deployment $(S^{(t)}+v, z(v))$.  The detailed algorithm for computing the best assignment if we add a candidate server is listed in Algorithm \ref{alg:assign}. Overall, the complete  algorithm for solving our SDU problem is listed in Algorithm \ref{alg:SDU}.

\SetKwComment{Comment}{/* }{ */}
\begin{algorithm}[t]
\footnotesize
\caption{$\mathsf{ASSIGN}$}
\label{alg:assign}
\Comment{Given   current deployment $(S,z)$, if we augment it with a new server $v$, find   best assignment $z^{*}$ for   augmented set $(S+v)$.}
\KwIn{$S, z, v$}
\KwOut{$z^{*}$}
\Begin{
		$I \gets \emptyset$\\
		$\Delta \gets 0$\\
		\While{$(I \neq [N])$}{
			$i \gets \arg \max \{ \Delta(j): {j \in [N]\setminus I}\}$\\ 
			\If{$(\Delta(i) \le 0)$}{
				\textbf{break}
			}
			$I \gets I+i$\\
			$\Delta \gets \Delta + \Delta(i)$ \\
		}

	$z^{*} \gets z$\\
	\If{$(\Delta > 0)$} { 
		\For{$(i \in I)$}{
			$z^{*}_{iv} \gets 1$\\
			$z^{*}_{is} \gets 0$ for the server $s$ where $z_{is} =1$
		}
	}
}
\end{algorithm}

\begin{algorithm}[t]
\footnotesize
\caption{Greedy algorithm for $\mathsf{SDU}$}
\label{alg:SDU}
\Comment{Assumption: input parameters $\mu_i$, $\kappa_s$, and $\nu_s$ are obtained from historical data}
 \KwIn{\\
 \Indp 
$[M]$: set of server candidates\\
$[N]$: set of cells\\
$k$: number of servers to deploy\\
 $\mu_i$: mean workload of each cell $i\in [N]$\\
$\kappa_s$: mean capacity of each server  $s \in [M]$\\
$\nu_s$: standard deviation of the difference between server $s$'s capacity and the total cell workload\\
 }
\KwOut{$(S^*, z^*) = \arg \max_{S,z} \Omega(S, z)$}
\Begin{
	$S_l, S_u, S \gets \emptyset$\\
	$z_l, z_u, z \gets 0$\\
	\For{$(t = 1 \rightarrow k)$} {
		$v_l \gets \arg \max \{ \Pi_l(S_l+v): v \in [M] \setminus S_l \} $\\
		$S_l \gets S_l + v_l$\\
		$z_l \gets \mathsf{ASSIGN}(S_l, z_l,  v_l)$
	}	
	\For{$(t = 1 \rightarrow k)$} {
		$v_u \gets \arg \max \{  \Pi_u(S_u+v): v \in [M] \setminus S_u\}$\\
		$S_u \gets S_u + v_u$\\
		$z_u \gets \mathsf{ASSIGN}(S_u, z_u,  v_u)$
	}

	$\hat{\Pi}_l \gets \hat{\Omega}(S_l, z_l)$, $\hat{\Pi}_u \gets \hat{\Omega}(S_u, z_u)$\\
	
	\If{$(\hat{\Pi}_l > \hat{\Pi}_u)$} {\Return $(S_l, z_l)$}
	\Return $(S_u, z_u)$
}
\end{algorithm}

 \textbf{Time complexity. } The time to run the $\mathsf{ASSIGN}$ algorithm (Algorithm \ref{alg:assign}) is $O(N^2)$ time because the time for line 5 to compute $\Delta$ is $O(1)$ and to find $i$ is $O(N)$. 
The time is  $O(|S| + N)$ to evaluate $\Pi_u(S)$ and $O(|S|N)$ to evaluate $\Pi_l(S)$. 
In the main algorithm, Algorithm \ref{alg:SDU}: 1) the time to execute line 5 totaling for all $k$ iterations is $O(MNk^2)$ and the time to execute line 7 totaling for all $k$ iterations is $O(kN^2)$. Therefore, the total time for lines 4-8 is $O(Nk(N+kM))$. This is also the time complexity for lines 9-13. Lines 14-21 take less time. Overall, the time complexity to compute the final server placement solution is $O(Nk(N+kM))$.

\textbf{Optimality analysis. } Thanks to the sandwiching of the objective function $\Pi(S)$ between submodular functions $\Pi_l(S)$ and $\Pi_u(S)$ and the $(1-1/e)$-bound of the standard greedy algorithm for submodular maximization, the near-optimality of our proposed algorithm is guaranteed accordingly. However, in our algorithm,  we approximate $\Pi(S_l)$ and $\Pi(S_u)$  with $\Omega(S_l, z_l)$ and $\Omega(S_u, z_u)$, respectively.  This approximation of course affects the theoretical bound. Therefore, our bound remains a heuristic. This is fair because approximation of $\Pi$ is a must due to the NP-hardness of  evaluating $\Pi$.

%% file: 1eval.tex
\section{Evaluation Study}\label{sec:evaluation} 
We conducted an  evaluation using a large cellular traffic dataset made available by  Chen et al. \cite{Modeling15-Chen}. This dataset provides hourly HTTP traffic statistics during 8 days for a cellular network of $N=13,296$ cells (base-stations) serving a medium-size city of China (50km x 60km).

\subsection{Setup}

The cellular network is geographically grid-partitioned into 16x16 = 256 regions. Some of these regions do not contain any cell.  Only 218 regions are non-empty. Therefore, we set the number of candidate servers to $M=218$ each located at the center of a non-empty region. We consider three cases of deploying the actual servers, $k \in \{10, 20, 30\}$.  
In the dataset, the position of a cell is a normalized longitude/latitude position which we convert  to a planar x/y position.  We use the Euclidean distance  between a cell $i$ and a server $s$ to represent their communication cost $d_{is}$ (normalized to $d_{is} \in [0, 1]$).

\begin{figure}[t]
\begin{center}
\subfigure[ Mean $\mu_i$]{\includegraphics[width=0.24\textwidth]{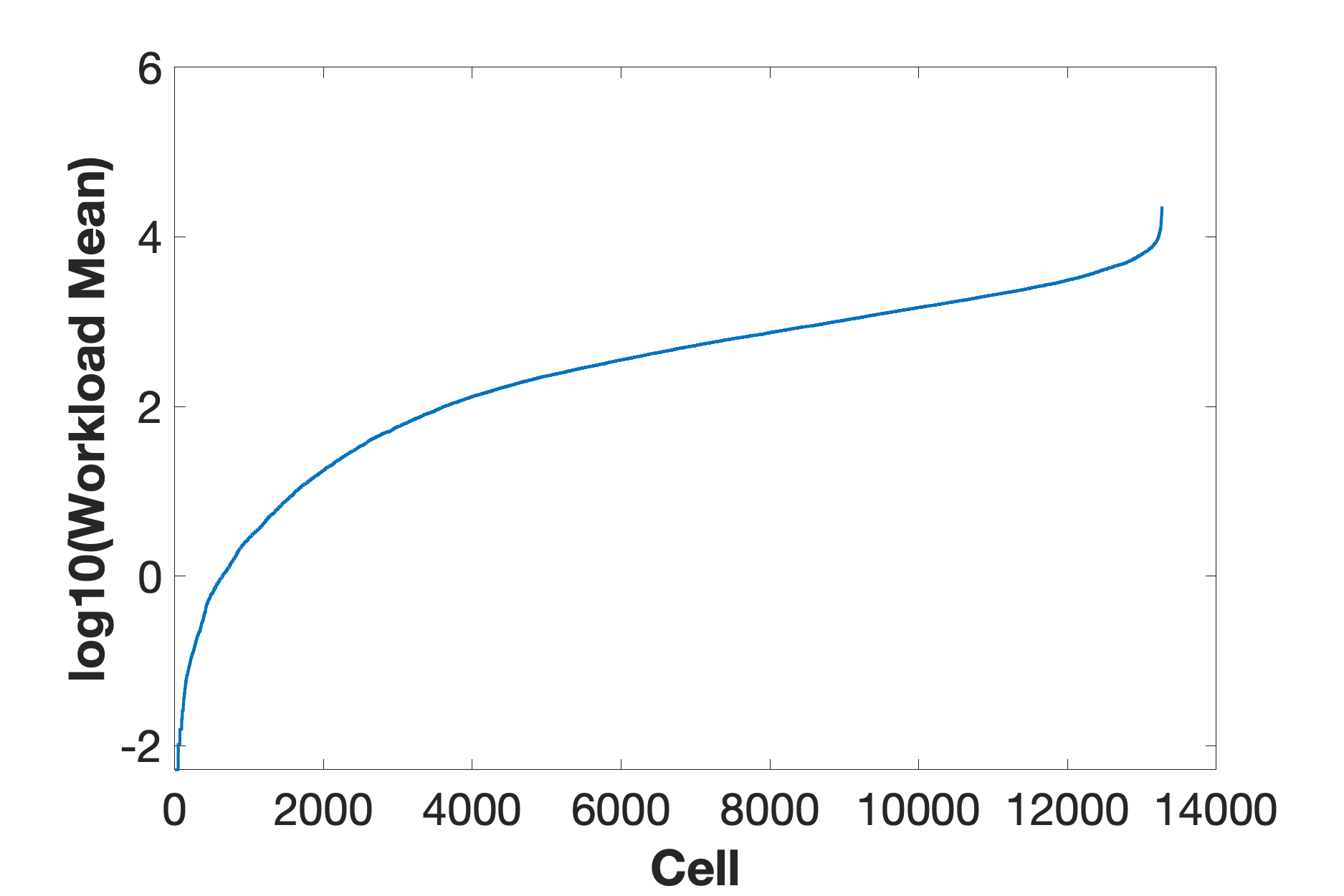}\label{vertex_weight_mean.jpg}}
\subfigure[ Coefficient of Variation $\frac{\sigma_i}{\mu_i}$]{\includegraphics[width=0.24\textwidth]{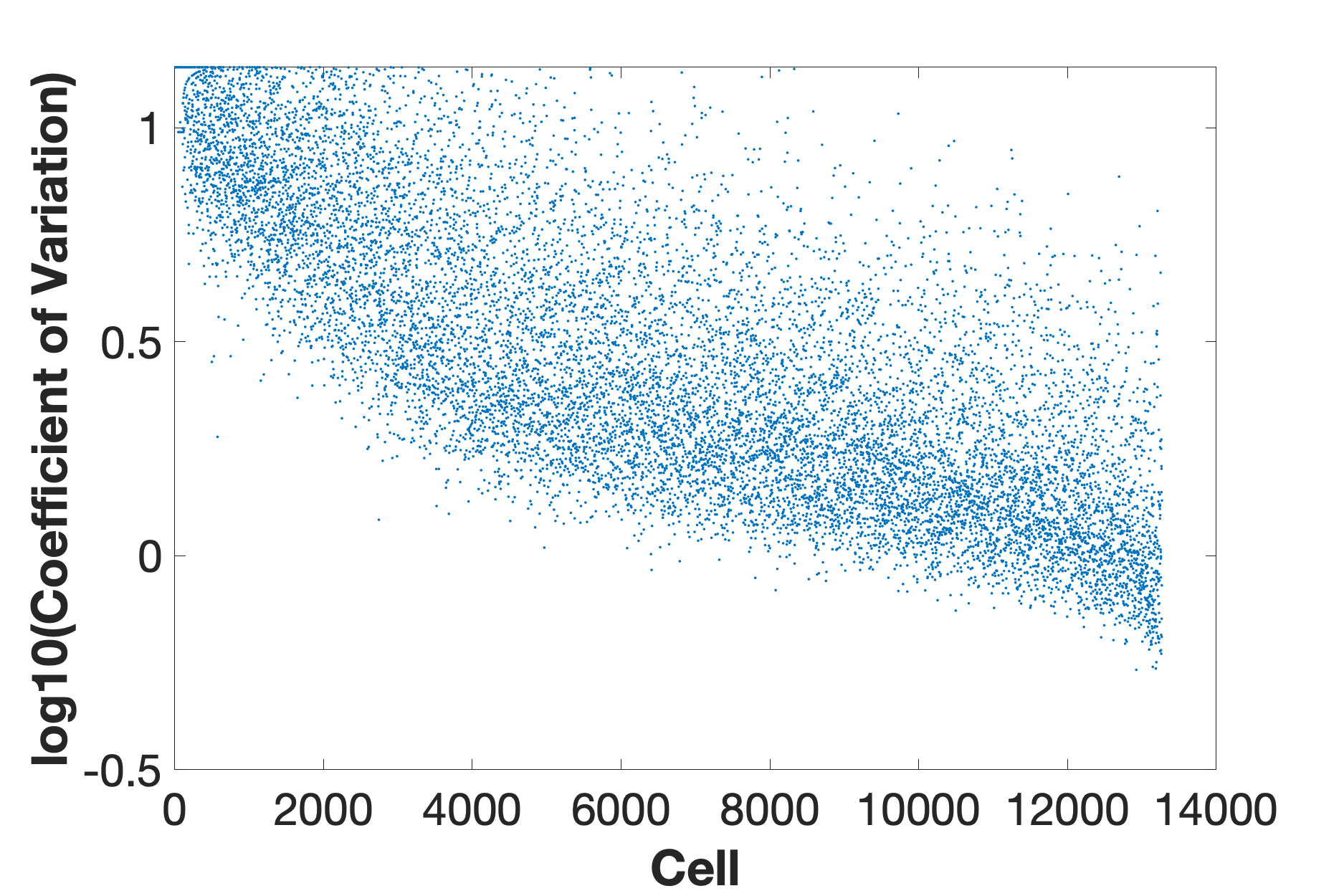}\label{vertex_weight_cv.jpg}}
\caption{Distribution of cell workload  over the time. Plots are in log-linear scale, where  the x-axis represents the cells in the increasing order of workload mean. (a) workload mean varies widely from cell to cell; (b) coefficient of variation varies from 0.8 (for ``busiest" cells) to 10 (for ``lightest" cells).} 
 \label{fig:vertexweight1}
 \end{center}
\end{figure}

\textbf{Cell workload. }
 The time-varying workload $W_i$ of each cell $i$ is   the number of transferred packets involving cell $i$. We collected this count at every hour during the entire 8-day period. There are thus 192  samples (8 days $\times$ 24 hours/day)  for each cell; these represent 192 times there is a workload demand from the cell. The empirical mean and variance of these samples for cell $i$ serve  as the mean $\mu_i$ and variance $\sigma_i^2$ for   $W_i$. Figure \ref{fig:vertexweight1} shows   mean $\mu_i$ and   coefficient of variation (cv) $\sigma_i/\mu_i$ for each cell $i$'s workload, which provides evidence of wide variation not only for the workload demand of a single cell over the time but also across different cells.
 
 \textbf{Server capacity. }  We assume that the capacity of a server changes hourly. In each hourly period,  the capacity of server $s$ is a random variable $K_s \sim \mathcal{N}(\kappa_s, \gamma_s^2)$,
  where 
  $  \kappa_s  = \mathcal{U}(1,  \kappa) \cdot \frac{1}{k} \sum_{i \in [N]} \mu_i,~ 
  \gamma_s = \gamma \cdot \kappa_s$.
We consider $\kappa \in \{0.7, 1.3\}$ (the higher, the more edge capacity with respect to the total workload demand from the cells) and $\gamma \in \{0.1, 0.9\}$ (the higher, the more time-varying for the capacity). Intuitively, for example, when $\kappa=0.7$ (low capacity) and $\gamma=0.9$ (high capacity variation), we have a setting where each server can carry somewhat $70\%\pm$ of the averaged-over-$k$-server workload  and each server's capacity varies within $90\%\pm$ from its mean. Given samples for $K_s$ and $W_i$ for each hourly period, we can obtain the corresponding $\nu_s$ in Eq \eqref{eq:nus} for each server.

\subsection{Results}
We compare  our  proposed $\mathsf{SDU}$ algorithm  (Algorithm \ref{alg:SDU})  to three baseline algorithms:
1) $\mathsf{RAND}$: $k$ servers are chosen uniformly at random and assigned uniformly at random to the cells; 
2) $\mathsf{FACILITY}$: $k$ servers are chosen and assigned to the cells according to the Facility Location algorithm, where communication cost is considered a kind of distance. This algorithm prioritizes communication efficiency only;
and
3) $\mathsf{KNAPSACK}$: $k$ servers with highest capacity are chosen, then assigned to the cells according to a Multiple-Knapsack algorithm. Only  computing efficiency is considered. This algorithm does not work with uncertainty, and so to be applicable, we replace random variables $X_i$, $K_s$ with their mean values $\mu_i$,  $\kappa_s$, respectively.

These algorithms are compared in terms of empirical efficiency - measured as the empirical average of objective function $\Pi$ when back-tested on  actual workloads/capacities for each hour of the 8-day, 198-hour, period. 
It is noted that objective $\Pi$ is a linear combination of two objectives that have different units, computing unit versus communication unit. As such,  setting coefficient $\lambda$ to 0.5 does not mean giving equal priority to each objective. For a fair balancing of these units, we applied a normalization method as follow. First, we generate a series of uniform-random server deployment $(S,z)$ and compute 
$a = {f}(S,z),~b = g(S,z)$,
which are the corresponding computing efficiency and communication efficiency, for each of the 192 hourly periods. Then, 
 $a,b$ are averaged over all samples. Finally, we use the following adjusted formula for $\Pi$ instead of the original formula:
$\Pi_{adjusted}(S) = \frac{\lambda}{   {a}} \cdot \mathbb{E}[f(S,\zeta(S)) + \frac{1-\lambda}{  {b}} \cdot g(S,\zeta(S))$.
With this, $\lambda=1/2$ more or less represents an equal priority between the two objectives. We run with $\lambda \in \{0.2, 0.5, 0.8\}$.

\begin{figure*}[t]
\begin{center}
{\includegraphics[width=0.32\textwidth]{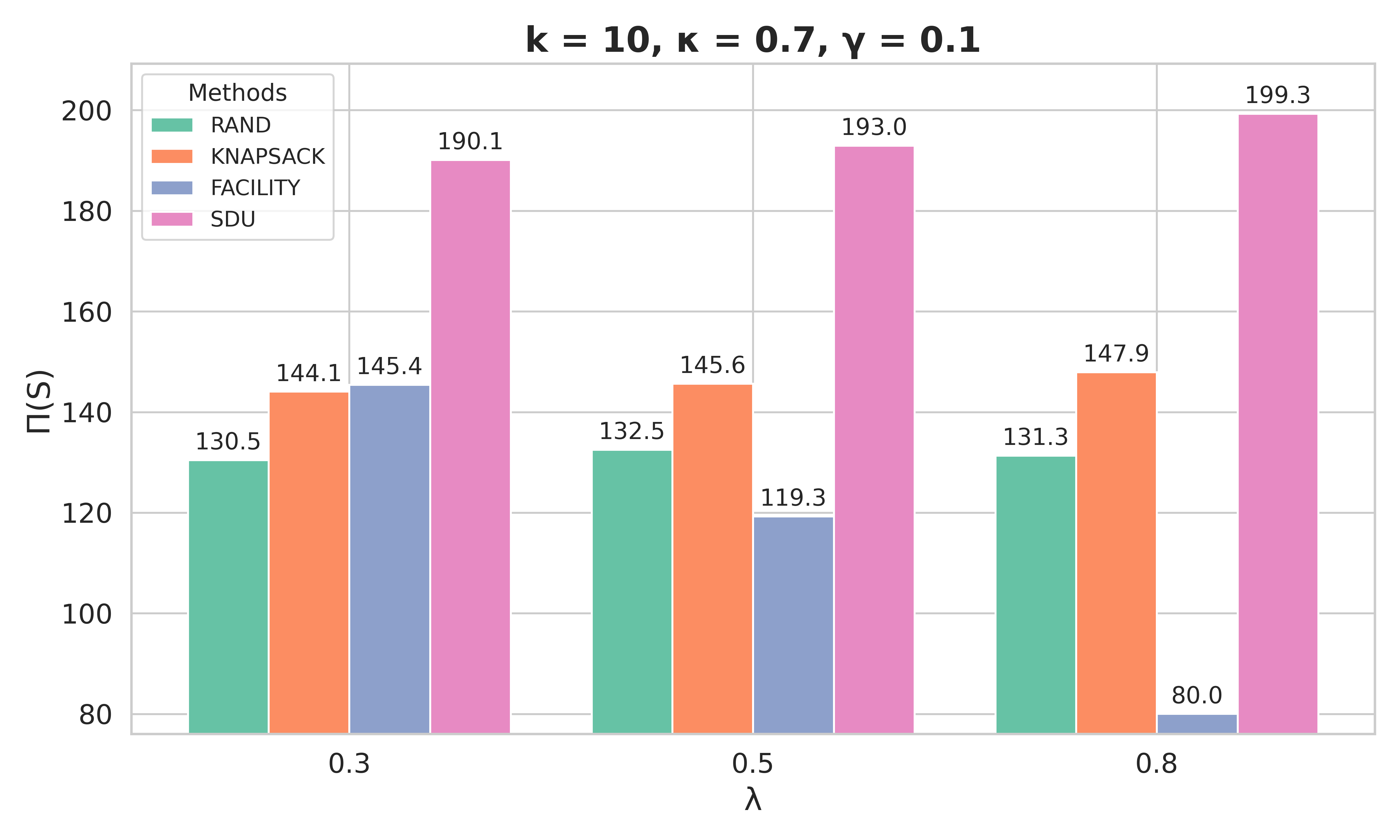}\label{fig:heatmap_2AM.png}}
{\includegraphics[width=0.32\textwidth]{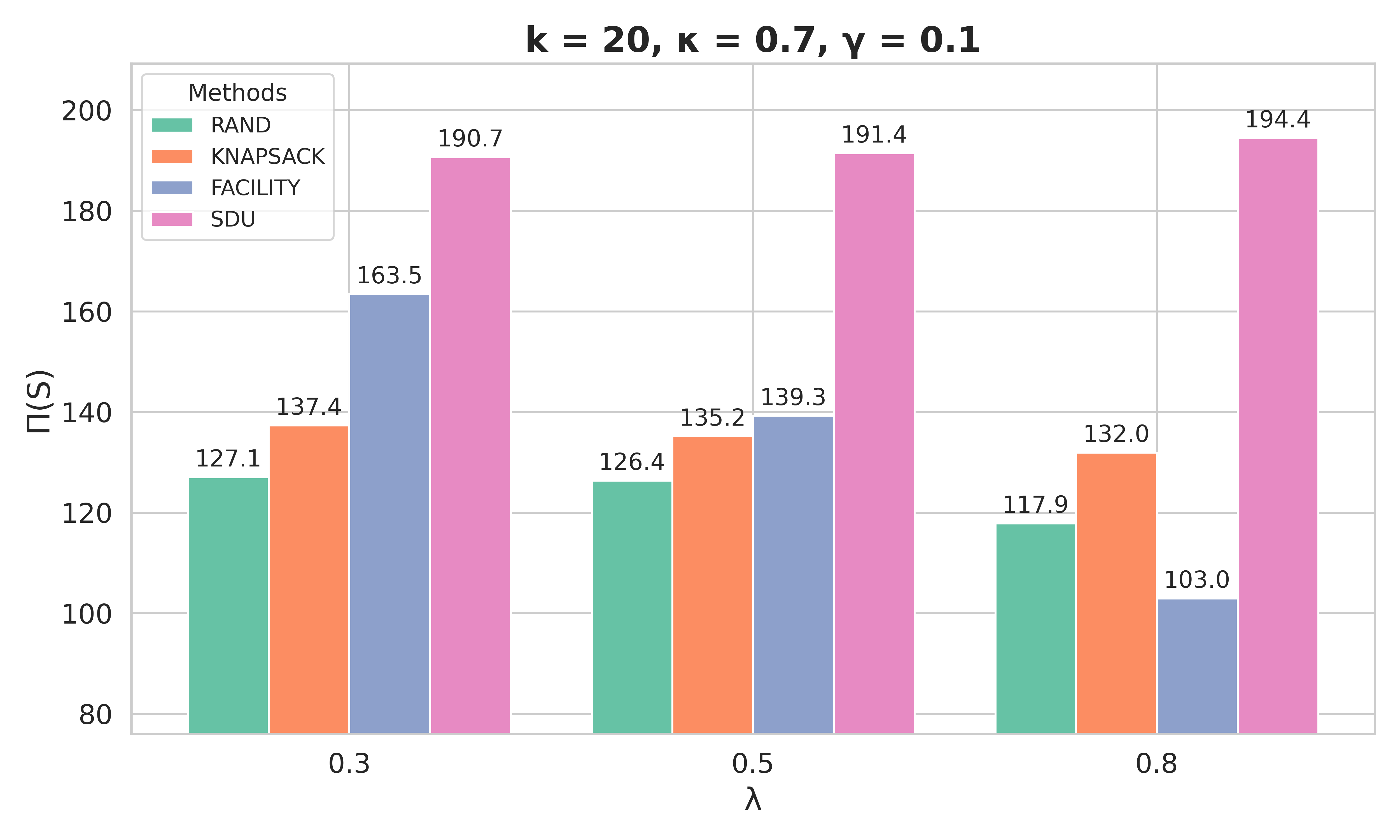}\label{fig:heatmap_2AM.png}}
{\includegraphics[width=0.32\textwidth]{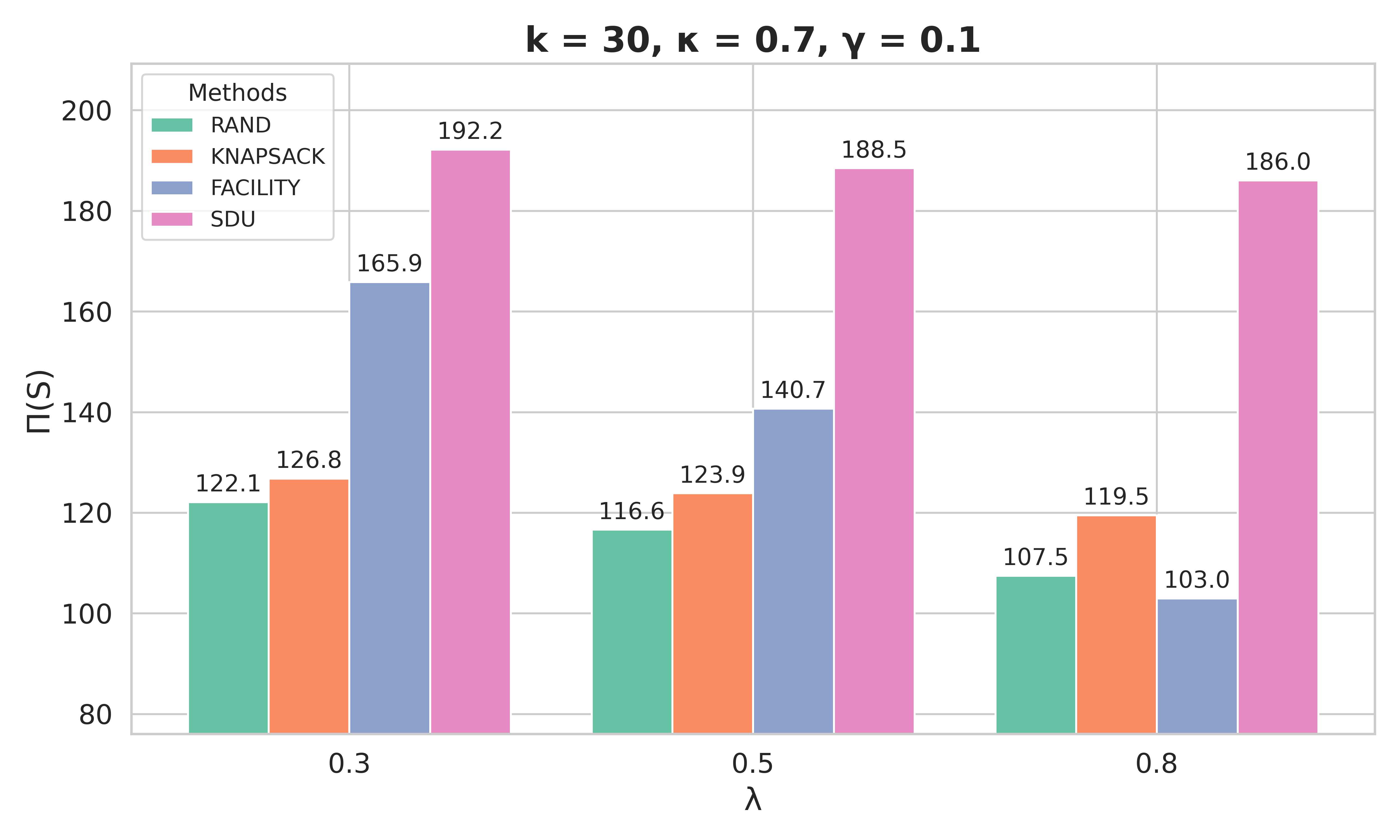}\label{fig:heatmap_2AM.png}}
{\includegraphics[width=0.32\textwidth]{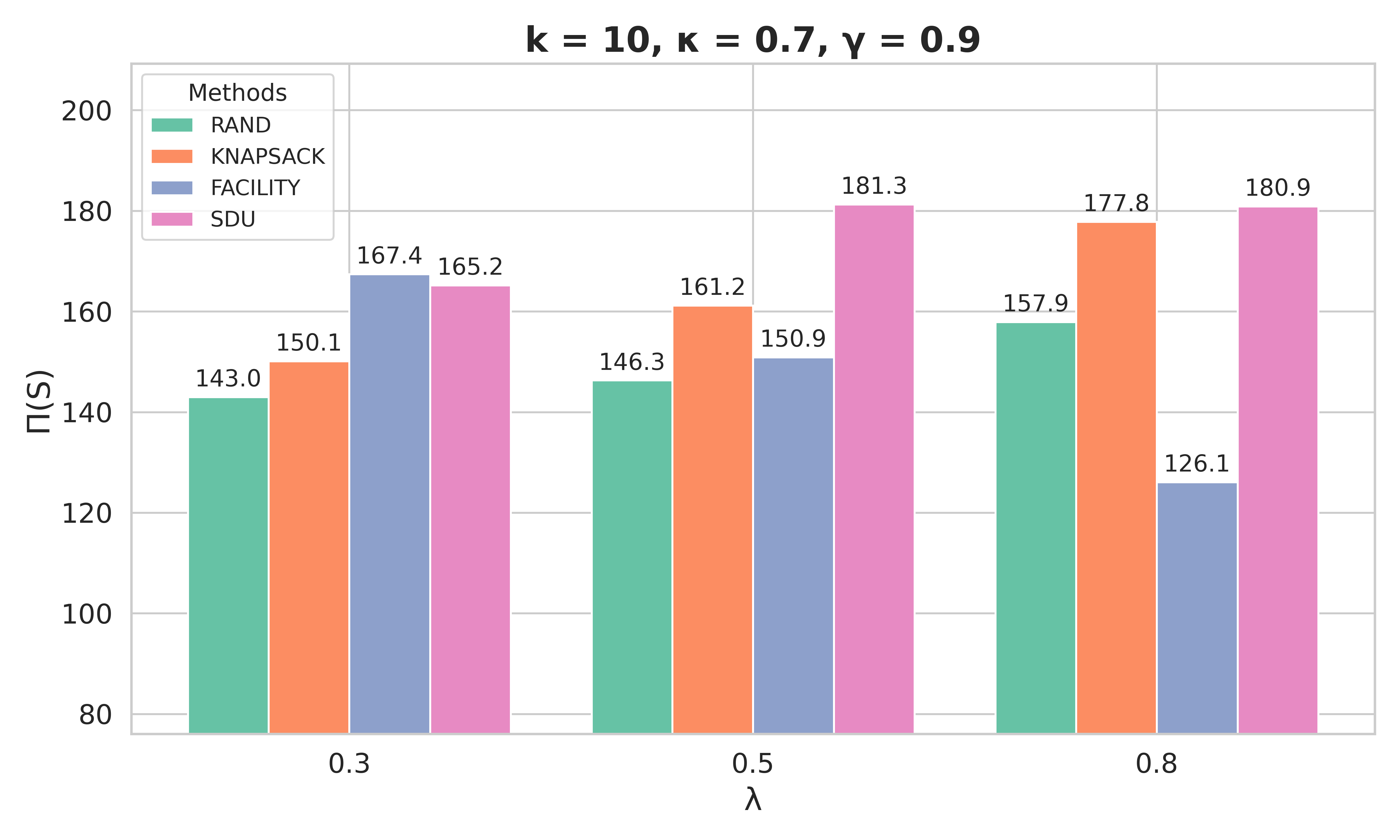}\label{fig:heatmap_4AM.png}}
{\includegraphics[width=0.32\textwidth]{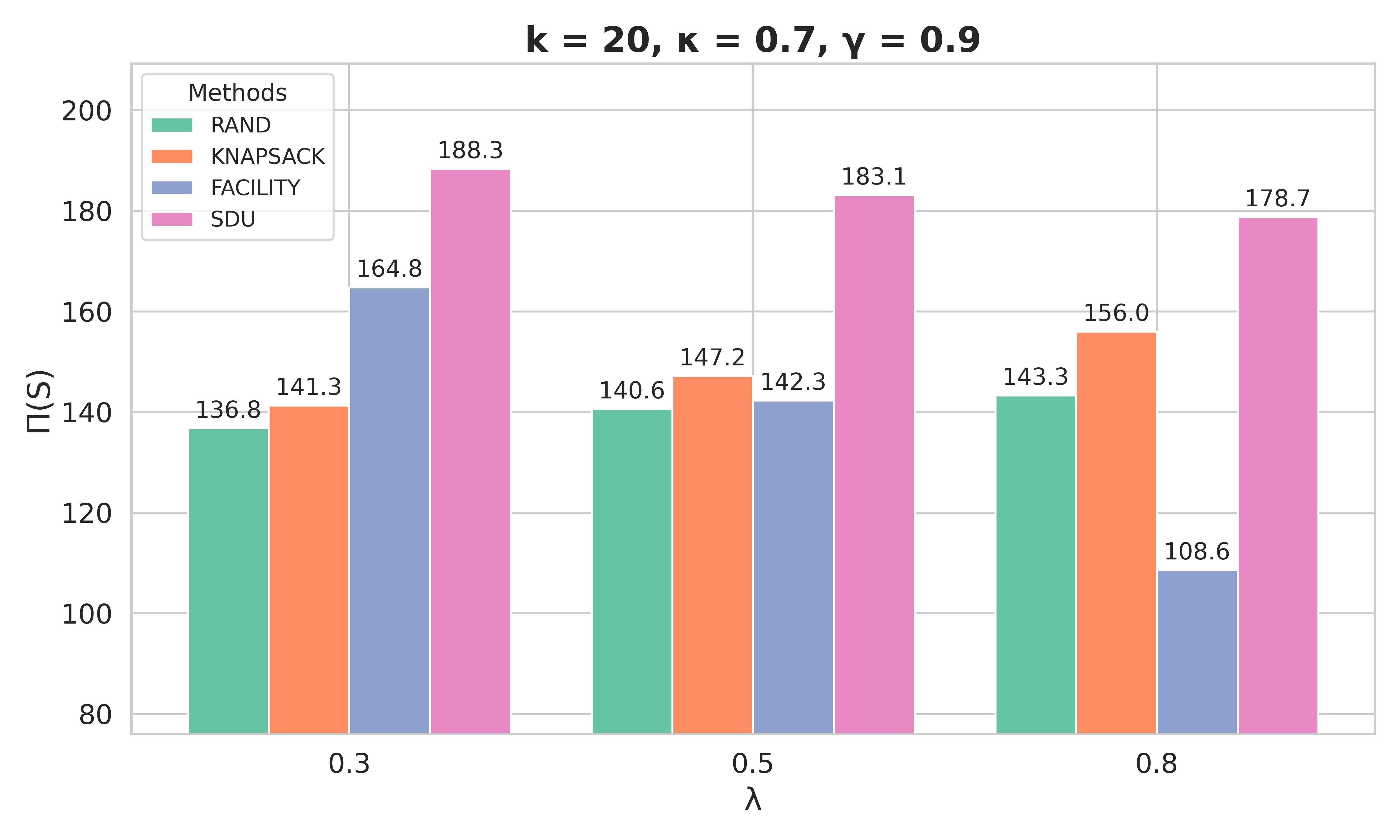}\label{fig:heatmap_4AM.png}}
{\includegraphics[width=0.32\textwidth]{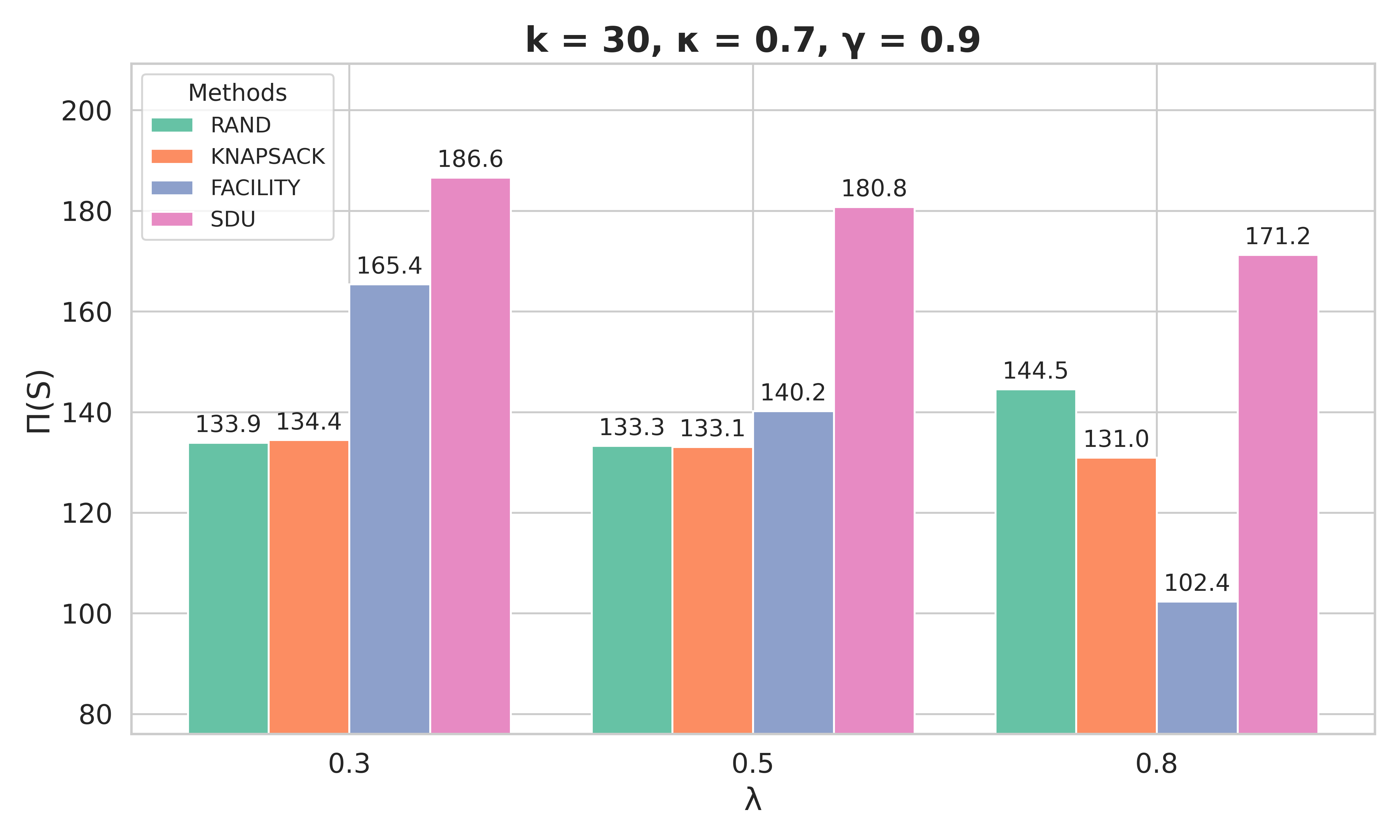}\label{fig:heatmap_4AM.png}}
\caption{Less edge server capacity ($\kappa=0.7$): comparison of $\mathsf{SDU}$ to other benchmark algorithms under various types of workload-capacity uncertainty.} 
 \label{fig:efficiency_lowedgecapacity}
 \end{center}
\end{figure*}

\begin{figure*}[t]
\begin{center}
\subfigure[Less edge server capacity: $\kappa=0.7$]{\includegraphics[width=0.48\textwidth]{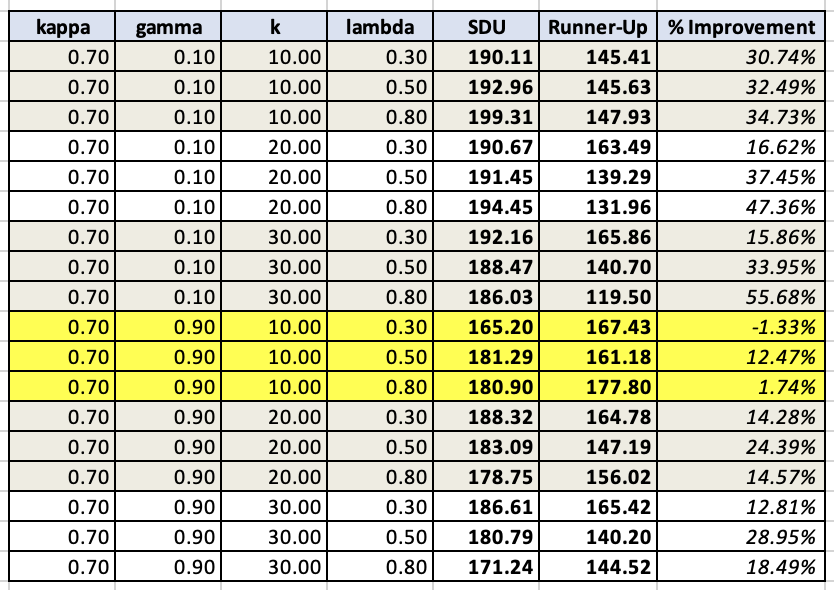}\label{fig:heatmap_2AM.png}}
\subfigure[More edge server capacity: $\kappa=1.3$]{\includegraphics[width=0.48\textwidth]{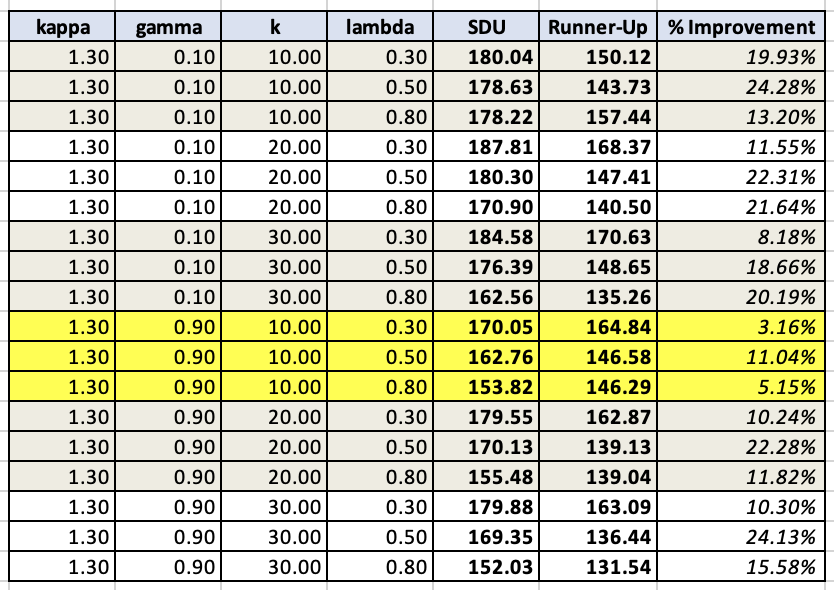}\label{fig:heatmap_2AM.png}}
\caption{Percentage improvement of $\mathsf{SDU}$ compared to the second-best algorithm.} 
 \label{fig:improvement}
 \end{center}
\end{figure*}

First, let us focus on the case where the edge server capacity is less ($\kappa=0.7$). The efficiency results are plotted in Figure \ref{fig:efficiency_lowedgecapacity}.  As expected, 1) $\mathsf{KNAPSACK}$ tends to be more efficient when $\lambda$ increases, i.e., prioritizing computing efficiency over communication efficiency; 2) $\mathsf{FACILITY}$ is the opposite direction;   3) $\mathsf{RAND}$ is agonistic to $\lambda$, which is explainable because this method is randomized regardless of any priority; and 4) most noticeably, however, the proposed algorithm, $\mathsf{SDU}$, is not only the best, but also outperforms all the other algorithms by substantial margins. Similar observations hold true for the case where the edge server capacity is less ($\kappa=1.3$), whose efficiency results are plotted in Figure \ref{fig:efficiency_highedgecapacity}.

Figure \ref{fig:improvement} shows the degree (percentage) of improvement when  $\mathsf{SDU}$ is compared to the second-best algorithm - which is $\mathsf{KNAPSACK}$ in the case of prioritizing computing efficiency (higher $\lambda$) or $\mathsf{FACILITY}$ in the case of prioritizing communication efficiency (smaller $\lambda$). The improvement can be as high as 55.68\% (when $\kappa=0.7, \gamma=0.1, k=30, \lambda=0.8$). On average over 36 settings, the improvement is 19.58\%.  The improvement tends to increase when the edge deploys less capacity, i.e., smaller $\kappa$.  
The least improvement is observed when the edge deploys fewest servers ($k=10$) that have highly-varying capacity ($\gamma=0.9$); this is highlighted with yellow color in Figure \ref{fig:improvement}. In this constrained setting, there is one outlier case, $\kappa=0.7, \gamma=0.9, k=10, \lambda=0.3$, where the proposed algorithm is slightly less efficient than $\mathsf{FACILITY}$, which is only by 1.33\%,
hence negligible. Overall, it is clear that $\mathsf{SDU}$ is far superior to all the benchmark algorithms.

\begin{figure*}[t]
\begin{center}
{\includegraphics[width=0.32\textwidth]{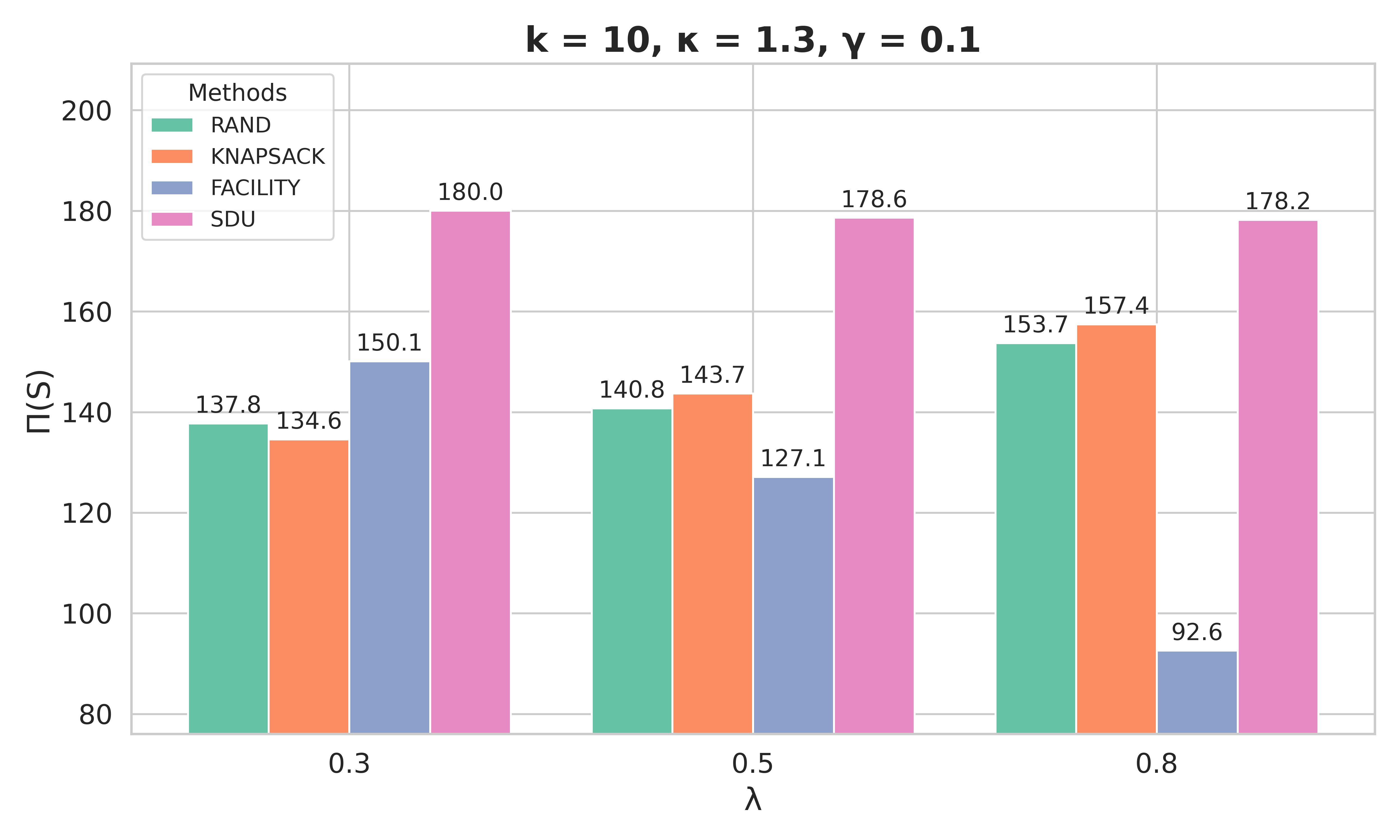}\label{fig:heatmap_2AM.png}}
{\includegraphics[width=0.32\textwidth]{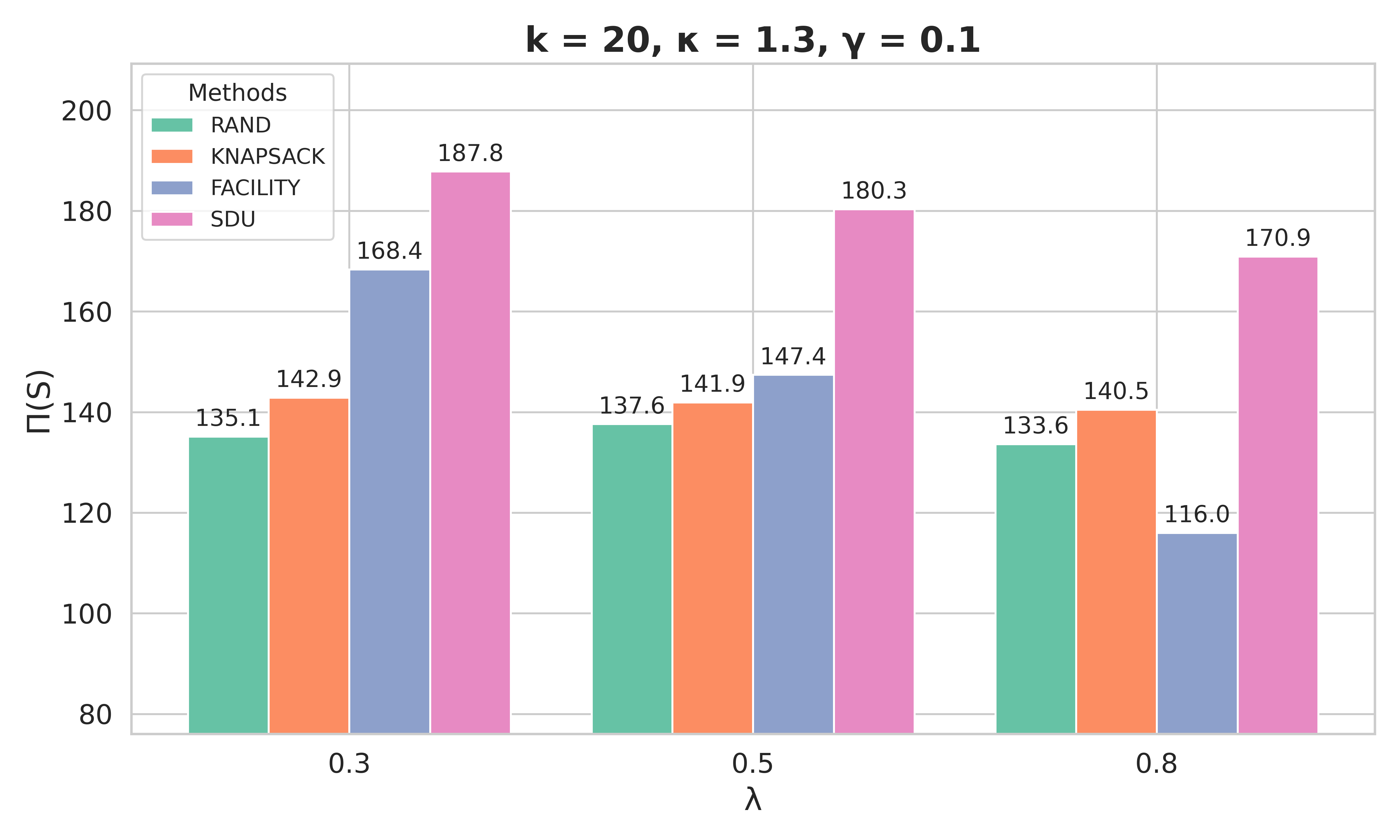}\label{fig:heatmap_2AM.png}}
{\includegraphics[width=0.32\textwidth]{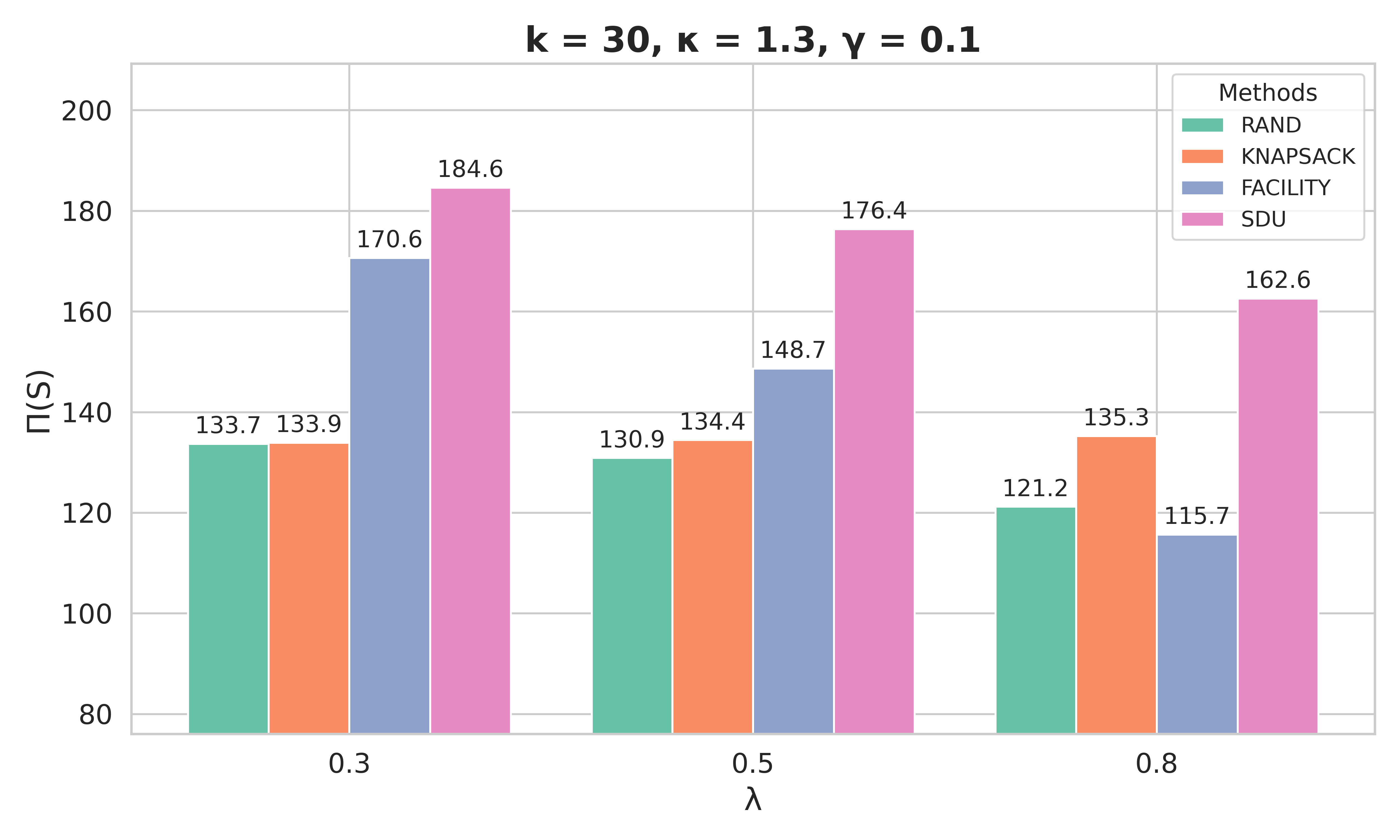}\label{fig:heatmap_2AM.png}}
{\includegraphics[width=0.32\textwidth]{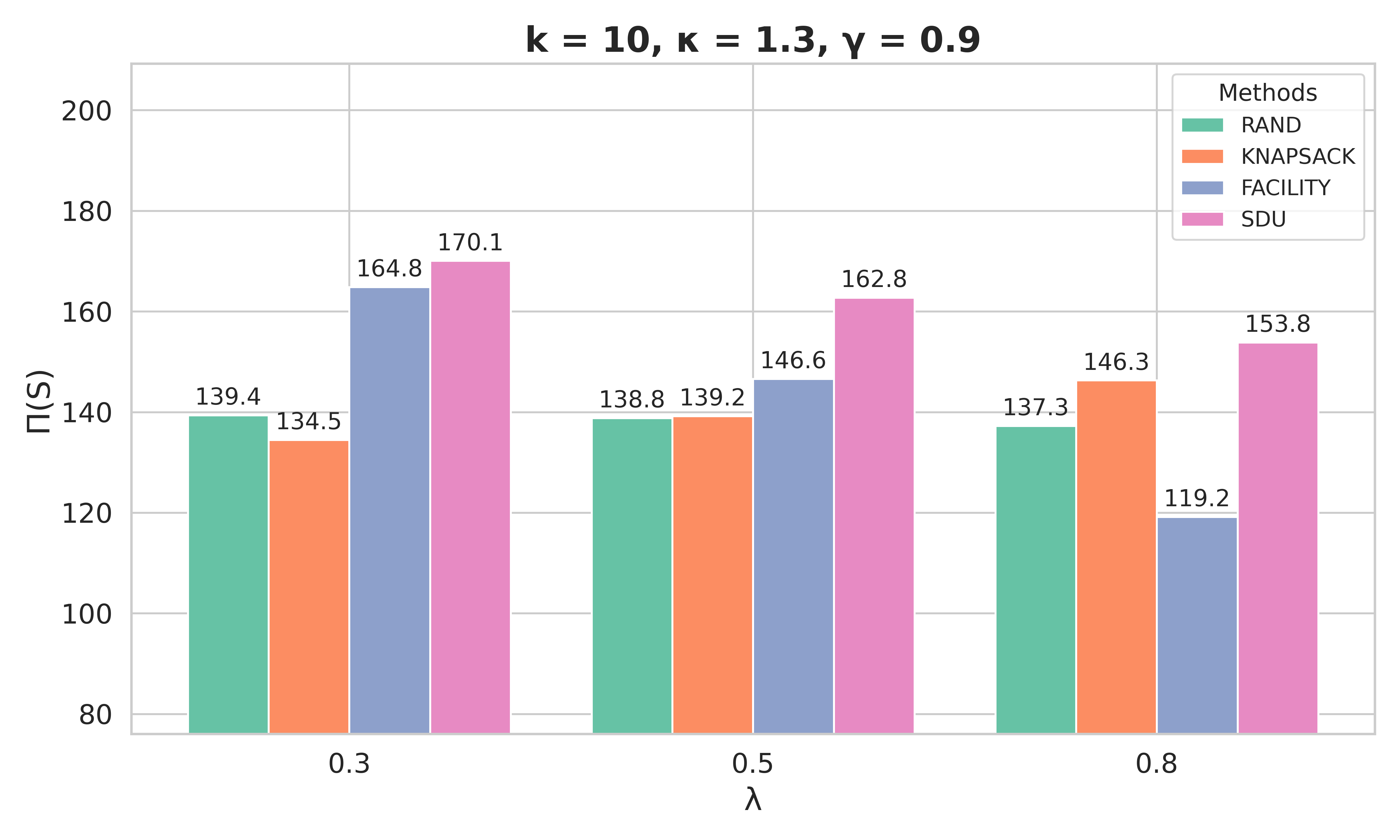}\label{fig:heatmap_4AM.png}}
{\includegraphics[width=0.32\textwidth]{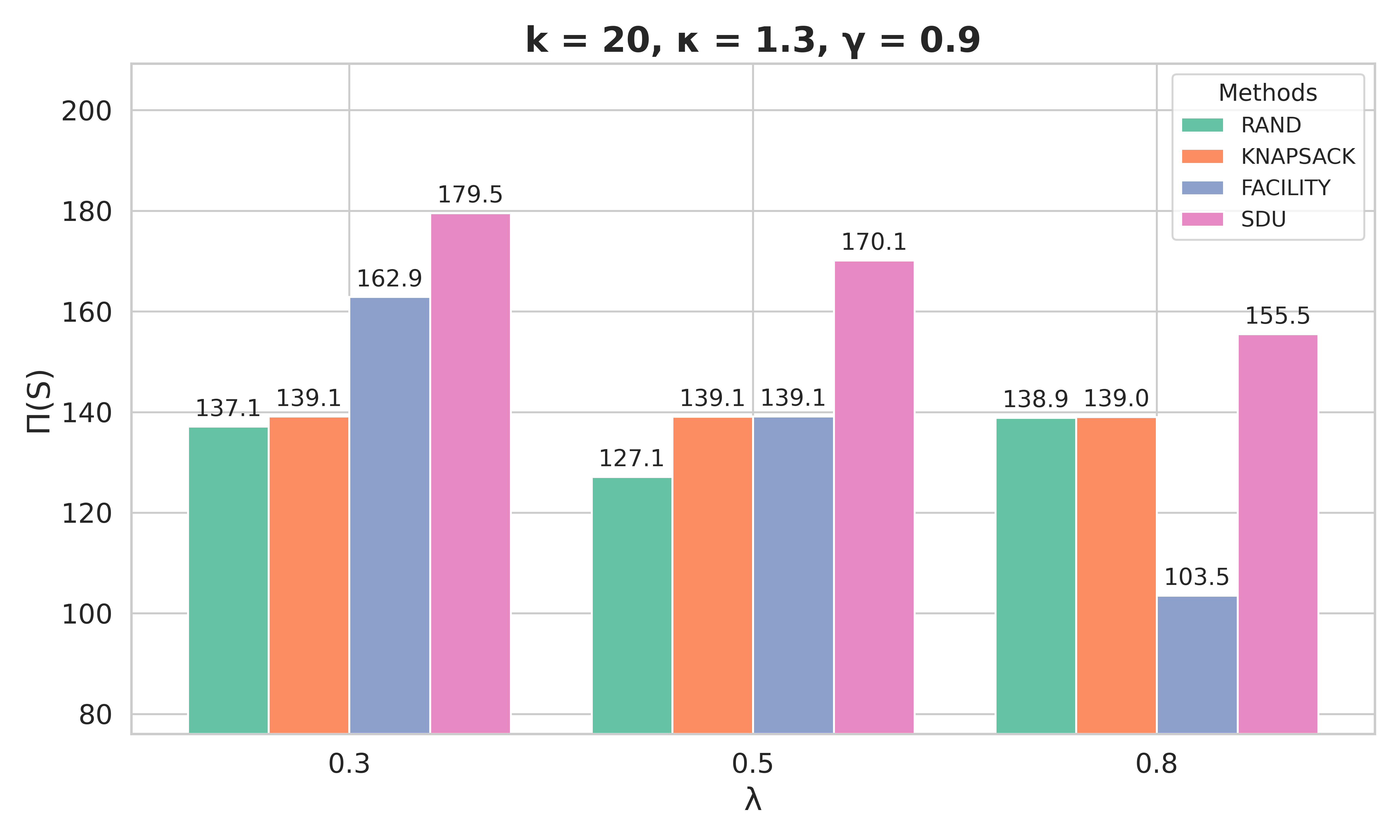}\label{fig:heatmap_4AM.png}}
{\includegraphics[width=0.32\textwidth]{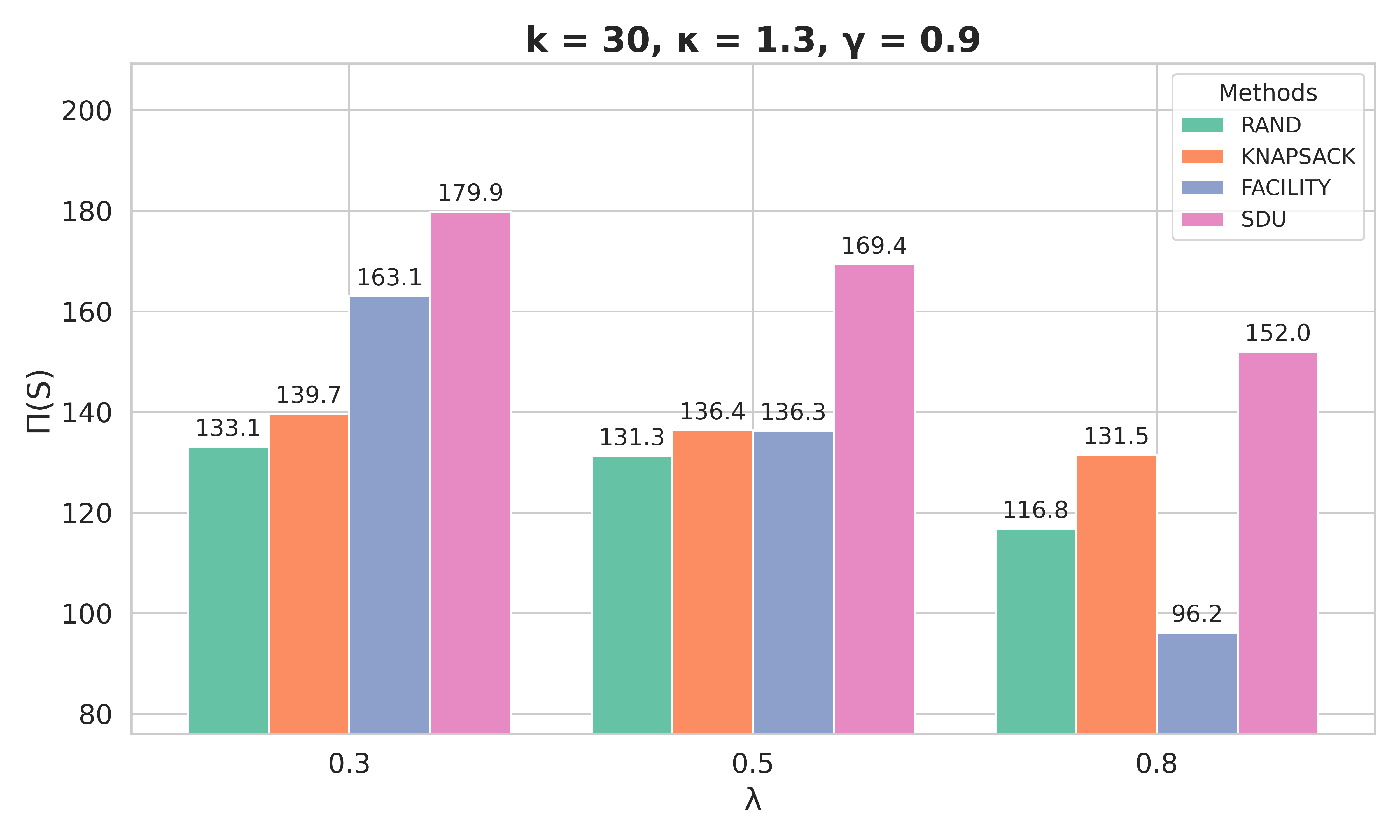}\label{fig:heatmap_4AM.png}}
\caption{More edge server capacity ($\kappa=1.3$): comparison of $\mathsf{SDU}$ to other benchmark algorithms under various types of workload-capacity uncertainty.} 
 \label{fig:efficiency_highedgecapacity}
 \end{center}
\end{figure*}

%% file: 1conclusions.tex
\section{Conclusions\label{sec:conclusions}}

Our research addresses a gap in the literature where no prior research on MEC server deployment has tackled long-term robustness against time-varying workloads and server capacities. We have proposed an original formulation of this problem as a stochastic bilevel-optimization problem, capturing workload and capacity uncertainties as random variables and optimizing two objectives, namely computing efficiency and communication efficiency. We propose a polynomial-time approximate algorithm that leverages submodular properties to sandwich the bilevel objective, achieving bounded approximality despite the problem's NP-(strongly) hard nature. Our evaluation study using real-world mobile traffic data has demonstrated that our approach outperforms intuitive methods by an average of 19\%, with improvements up to 55\%, and supports flexible prioritization of computing and communication efficiencies. Future work will explore extending the model to incorporate dynamic server scaling and adaptive workload prediction, further enhancing its real-world applicability.